\documentclass[11pt, oneside]{article} 
\pdfoutput=1

\usepackage{jheppub}

\usepackage{amsmath,amssymb,amsfonts,amsthm}
\usepackage{color}
\usepackage{booktabs}
\definecolor{darkblue}{rgb}{0.1,0.1,.7}
\usepackage[]{amsmath}
\usepackage[]{graphicx}
\usepackage[]{latexsym}
\usepackage{amscd}
\usepackage[all,cmtip]{xy}
\usepackage{mathrsfs}
\usepackage{bbold}
\usepackage{subfigure}
\usepackage[margin=10pt,font=small,labelfont=bf]{caption}
\usepackage{simplewick}
\usepackage{changepage}
\usepackage[]{algorithm2e}
\usepackage{booktabs,multirow}
\usepackage[OT2,OT1]{fontenc}

\newcommand{\assign}{:=}
\newcommand{\backassign}{=:}

\newcommand{\nospace}{}

\newtheorem{theorem}{Theorem}[section]
\newtheorem{lemma}[theorem]{Lemma}
\newtheorem{proposition}[theorem]{Proposition}
\theoremstyle{remark}
\newtheorem{remark}[theorem]{Remark}
\makeatletter
\def\@fpheader{\ }
\makeatother

\title{Twist accumulation in conformal field theory.\\A rigorous approach to
	the lightcone bootstrap}
\author{Sridip Pal$^{a,b}$, Jiaxin Qiao$^{c,d,e}$, Slava Rychkov$^{e}$}
\affiliation{$^a$Walter Burke Institute for Theoretical Physics,  California Institute of Technology,  Pasadena, CA, USA  \\ $^b$School of Natural Sciences, Institute for Advanced Study, Princeton, NJ 08540, USA\\ $^c$Laboratory for Theoretical Fundamental Physics, Institute of Physics, École Polytechnique Fédérale de Lausanne (EPFL), CH-1015 Lausanne, Switzerland\\ 
$^d$Laboratoire de Physique de l'École normale supérieure, ENS, Université PSL, CNRS, Sorbonne Université, Université de Paris, F-75005 Paris, France\\ 
	 $^e$Institut des Hautes \'Etudes Scientifiques, 91440 Bures-sur-Yvette, France\\
	}

\abstract{We prove that in any unitary CFT, a twist gap in the spectrum of operator
	product expansion (OPE) of identical scalar quasiprimary operators (i.e.\,$\phi
	\times \phi$) implies the existence of a family of quasiprimary operators
	$\mathcal{O}_{\tau, \ell}$ with spins $\ell \rightarrow \infty$ and twists
	$\tau \rightarrow 2 \Delta_{\phi}$ in the same OPE spectrum. A similar
	twist-accumulation result is proven for any two-dimensional
	Virasoro-invariant, modular-invariant, unitary CFT with a normalizable
	vacuum and central charge $c > 1$, where we show that a twist gap in the
	spectrum of Virasoro primaries implies the existence of a family of Virasoro
	primaries $\mathcal{O}_{h, \bar{h}}$ with $h \rightarrow \infty$ and
	$\bar{h} \rightarrow \frac{c - 1}{24}$ (the same is true with $h$ and
	$\bar{h}$ interchanged). We summarize the similarity of the two problems and
	propose a general formulation of the lightcone bootstrap. }

\begin{document}

\maketitle

\section{Introduction}

In this paper we will discuss conformally invariant unitary theories in $d
\geqslant 2$ dimensions, in Lorentzian signature. In this context, an
important quantity characterizing any local operator is its twist $\tau =
\Delta - \ell$ where $\Delta, \ell$ are the dimension and spin. We will be
considering theories with a twist gap, i.e.~theories in which all operators
have twist $\tau \geqslant \tau_{\rm gap} > 0$.

It has been noticed that a twist gap in the spectrum implies the existence of
twist-accumulating families of operators, i.e.~families of quasiprimaries
$\mathcal{O}_{\tau,\ell}$ whose spin $\ell \rightarrow \infty$ while twist
$\tau$ approaches a finite limit $\tau_{\infty}$. Namely, globally
conformally invariant theories in $d \geqslant 2$ with a twist gap should
contain one twist-accumulating theory of quasiprimaries for every scalar
quasiprimary $\phi$, with $\tau_{\infty} = 2 \Delta_{\phi}$ \cite{Fitzpatrick:2012yx,Komargodski:2012ek}. 
This was argued to follow from
the crossing symmetry of the four-point (4pt) function $\langle \phi \phi
\phi \phi \rangle$ in a certain ``lightcone limit''---an argument known
as ``lightcone bootstrap''. However, as we review below, the existing
derivations have various caveats, which may raise doubts in the universal
validity of this result. One goal of this work is to remove this doubt. We
will show a way to do the lightcone bootstrap in a mathematically rigorous
way.\footnote{Historically, twist accumulation of high-spin operators was first studied in 1973 in perturbation theory by Parisi \cite{Parisi:1973xn}, and by Callan and Gross \cite{Callan:1973pu}.}

There is a related twist-accumulation result for the two-dimensional (2D) Virasoro-invariant
unitary CFTs with $c > 1$ and a twist gap in the spectrum of
Virasoro primaries. In particular, such theories should have a
twist-accumulating family of Virasoro primaries with $\tau_{\infty} = (c - 1)
/ 12$ {\cite{Collier:2016cls,Afkhami-Jeddi:2017idc,Benjamin:2019stq}}. This
should follow from the modular invariance of the CFT partition function, in a
``lightcone limit'' of modular parameters. Once again, a rigorous derivation
has been lacking. We will provide it here, by adapting our new lightcone
bootstrap argument to the modular bootstrap case.

We will start with the lightcone bootstrap case in Section \ref{LC}: we
present the intuitive argument, explain the caveats, and then give a rigorous
proof, based on the Euclidean CFT axioms as stated e.g.~in \cite{Kravchuk:2021kwe}. We then treat the modular bootstrap case in Section \ref{LMB}. Here we
start with a rigorous proof and then discuss the intuitive arguments from the
prior literature and explain the caveats. The lightcone and modular sections
are written so that they can be read independently.
Similarities between the two problems are emphasized in Section \ref{GF}. In
Section \ref{Virasoro} we review the conjectures in 2D CFT coming from
analyzing the 4pt function using the Virasoro fusion kernel and crossing
symmetry, as opposed to the modular invariance of the partition function. We then
conclude and outline directions for further work. Some of the more involved
computations are delegated to the appendices.

\section{Lightcone bootstrap}\label{LC}

In a unitary CFT in $d \geqslant 2$ dimensions, we consider the 4pt function
$g (z, \bar{z})$ of a Hermitian scalar quasiprimary\footnote{In this paper we use the word ``quasiprimary'' to denote primary w.r.t. the global conformal group, even in $d>2$ where the word primary is often used for this purpose.} operator $\phi$ of dimension
$\Delta_{\phi}$, satisfying the crossing equation:
\begin{equation}
	g (z, \bar{z}) = \left( \frac{z \bar{z}}{(1 - z) (1 - \bar{z})}
	\right)^{\Delta_{\phi}} g (1 - z, 1 - \bar{z})\,. \label{crosseq}
\end{equation}
The function $g (z, \bar{z})$ is expanded in conformal blocks $g_{\tau, \ell}
(z, \bar{z})$ of the global conformal group:
\begin{equation}
	g (z, \bar{z}) = 1 + \sum_{\tau, \ell}\, p_{\tau, \ell}\, g_{\tau, \ell} (z,
	\bar{z}), \label{CBexp}
\end{equation}
where 1 is the contribution of the unit operator, and $\tau$ and $\ell$ are
the twist and the spin of the nontrivial quasiprimaries $\mathcal{O}_{\tau,
	\ell}$ in the OPE $\phi \times \phi$. The coefficients $p_{\tau, \ell}
\geqslant 0$ from unitarity. 

The 4pt function $g (z, \bar{z})$ is finite for $0\leqslant z,\bar{z}<1$, and the expansion \eqref{CBexp} converges in this region. This can be shown rigorously from the Euclidean CFT axioms \cite{Fitzpatrick:2012yx,Kravchuk:2021kwe}.

Conformal blocks have an expansion of the form
\begin{equation}
	g_{\tau, \ell} (z, \bar{z}) = (z \bar{z})^{\tau / 2} \sum_{n, m \in
		\mathbb{Z}_{\geqslant 0}} a_{n, m} z^n \bar{z}^m, \label{CBindexp}
\end{equation}
where $a_{n, m} \geqslant 0$ depend on $\tau, \ell, d$ \cite{Pappadopulo:2012jk}\cite{Fitzpatrick:2012yx}\cite{Kravchuk:2021kwe}.

We impose the twist gap assumption
\begin{equation}
	\tau \geqslant \tau_{\rm gap} > 0 \label{TGA}
\end{equation}
for all blocks appearing in {\eqref{CBexp}}.\footnote{Operators of twist exactly equal or close to $\tau_{\rm gap}$ need not exist. In other words, we do not assume that $\tau_{\rm gap}$ is the \emph{maximal} twist gap.} For $d > 2$, this assumption is
satisfied due to the unitarity bounds. For $d = 2$, the twist gap assumption
may be satisfied by theories without a local conserved stress tensor, such as
conformal defects or long-range fixed points. The twist gap assumption on
quasiprimaries is not satisfied in local 2D CFTs, because
infinitely many Virasoro descendants of the unit operator have twist 0 (the
stress tensor being the first of them), so this case is excluded from the
present discussion.\footnote{It is interesting to inquire what happens in
   local 2D CFTs satisfying the twist gap assumption on the
	spectrum of Virasoro primaries {\cite{Kusuki:2018wpa,Collier:2018exn}}. This
	is not the focus of our work but it will be discussed briefly in Section
	\ref{Virasoro}.}

Under the above assumptions, Refs.\,{\cite{Fitzpatrick:2012yx,Komargodski:2012ek}} argued that among the operators
$\mathcal{O}_{\tau, \ell}$ contributing to {\eqref{CBexp}}, there must be a
series of operators whose twist converges to $2 \Delta_{\phi}$ while spin
$\ell \rightarrow \infty$.

\subsection{Intuitive argument}\label{intuitive}

Here we will present an intuitive argument for the twist accumulation result
of {\cite{Fitzpatrick:2012yx,Komargodski:2012ek}}. The argument uses the CFT
crossing equation in the range of real $0 < z, \bar{z} < 1$, corresponding to
the Lorentzian signature. Furthermore, one considers the limit $z \rightarrow
0$, $\bar{z} \rightarrow 1$. This is known as the double lightcone
(DLC) limit, as it corresponds to one point in the 4pt function approaching
the lightcones of two other points.

In the DLC limit, the \textit{direct channel} (l.h.s. of Eq.\,{\eqref{crosseq}}) is expected to be dominated by the unit operator which by the
twist gap assumption is the only operator of zero twist:
\begin{equation}
	g (z, \bar{z}) \approx 1 \qquad \text{in the DLC limit} \label{unitdom}
	\quad (?),
\end{equation}
where (?) reminds that this is a non-rigorous statement and the rigorous result
will be more subtle, see Section \ref{direct}.

To motivate {\eqref{unitdom}}, one notes that each nontrivial conformal block
in the r.h.s.~of {\eqref{CBexp}} is \ $O (z^{\tau / 2})$, thus goes to zero as
$z \rightarrow 0$ for fixed $\bar{z}$. One then hopes that this remains true:
\begin{eqnarray}
	&  & \text{- when summed over all blocks;}  \label{hopeA}\\
	&  & \text{- not only for fixed $\bar{z}$ but as $\bar{z} \rightarrow 1$,
		i.e.~in the DLC limit} .  \label{hopeB}
\end{eqnarray}
Assuming {\eqref{unitdom}}, one then asks how this can be reproduced by the
\textit{crossed channel} (r.h.s.~of Eq.\,{\eqref{crosseq}}). One studies
how crossed-channel conformal blocks behave in the DLC limit. It can be seen
that:
\begin{equation}
	g_{\tau, \ell} (1 - z, 1 - \bar{z}) \approx (1 - \bar{z})^{\tau / 2}
	F_{\tau, \ell} (1 - z) \qquad \text{in the DLC limit} \label{CBcrossed},
\end{equation}
where the function $F_{\tau, \ell} (1 - z)$ blows up at most logarithmically
as $z \rightarrow 0$. Since this logarithmic growth cannot compensate the
factor $z^{\Delta_{\phi}}$ in the r.h.s.~of {\eqref{crosseq}}, one
concludes that a finite number of crossed conformal blocks cannot reproduce
{\eqref{unitdom}}. Rather, one needs a cooperative effect from an infinite
number of blocks. One notes that the needed asymptotics can be achieved by an
infinite family of operators $\mathcal{O}_{\tau, \ell}$ with
\begin{eqnarray}
	\tau = 2 \Delta_{\phi,} &  &  \label{condQ1}\\
	\sum_{\ell} p_{\tau, \ell} F_{\tau, \ell} (1 - z) & \sim & 1 /
	z^{\Delta_{\phi}} \qquad (z \rightarrow 0) .  \label{condQ2}
\end{eqnarray}
Condition {\eqref{condQ1}} allows the factor $ (1 - \bar{z})^{\tau / 2}$ in
{\eqref{CBcrossed}} to cancel $(1 - \bar{z})^{\Delta_{\phi}}$ in the
denominator of {\eqref{crosseq}}. Condition {\eqref{condQ2}} then allows the
total contribution of this family of operators to cancel
$z^{\Delta_{\phi}}$ in the numerator of {\eqref{crosseq}}. Condition
{\eqref{condQ1}} does not have to be satisfied exactly but only asymptotically
as $\ell \rightarrow \infty$.

This completes our review of the intuitive argument why an infinite family of
operators with twist accumulating to $2 \Delta_{\phi}$ may be expected to
exist.\footnote{For completeness, we mention two further ideas which were used
	in the past to argue for the twist accumulation. These ideas appear even more
	nontrivial to make rigorous and they will not play a role in our work. (1) Ref.~{\cite{Fitzpatrick:2012yx}} phrased its argument in terms of a spectral
	density in the twist space, which is a linear functional acting by integrating
	against continuous functions $f (\tau)$, and defined as a $z \rightarrow 0$
	limit of a family of functionals. However {\cite{Fitzpatrick:2012yx}} did not provide a
	proof that the limit exists, and this appears quite nontrivial. See the
	discussion in {\cite{Qiao:2017xif}}, footnote 26. (2) Ref.\,{\cite{Komargodski:2012ek}}, Section 3.5, gave an argument following the ideas
	of Alday-Maldacena {\cite{Alday:2007mf}} which view the CFT 4pt function as a
	4pt function in a massive 2D theory. } Now let us discuss the caveats.

The first caveat concerns Eq.\,{\eqref{unitdom}}. Passages in {\eqref{hopeA}},
{\eqref{hopeB}} clearly require justification, absent in the existing
literature. To have a chance for {\eqref{hopeB}}, it is clear that $\bar{z}$
should not approach 1 too fast, compared to how $z$ approaches $0$. This
suggests that DLC has to be formulated more precisely than merely $z
\rightarrow 0$, \ $\bar{z} \rightarrow 1$.

The second caveat concerns relying on Eq.\,{\eqref{CBcrossed}} to conclude that
only blocks with $\tau$ near $2 \Delta_{\phi}$ have good $\bar{z}$
dependence with respect to the $(1 - \bar{z})^{\Delta_{\phi}}$ factor in
the denominator of {\eqref{crosseq}}, to do the job of reproducing
{\eqref{unitdom}}. Indeed Eq.\,{\eqref{CBcrossed}} is for a single block. What
will happen when we sum infinitely many blocks? Can the blocks with $\tau$ not
near $2 \Delta_{\phi}$ somehow conspire to reproduce the needed behavior?
One needs to rule out this scenario before concluding that {\eqref{condQ1}},
{\eqref{condQ2}} is the only way of satisfying {\eqref{unitdom}}.

The third caveat concerns Eq.\,{\eqref{CBcrossed}} itself. It is easy to see
that this equation holds when $1 - \bar{z} \rightarrow 0$ for a fixed $z$, but
it's harder to control what happens in the DLC limit, except in even $d$ where
explicit formulas for $g_{\tau, \ell}$ are available. We have that
\begin{equation}
	g_{\tau, \ell} (1 - z, 1 - \bar{z}) = (1 - \bar{z})^{\tau / 2} F_{\tau,
		\ell} (1 - z) + O ( (1 - \bar{z})^{(\tau + 1) / 2}), \label{CBcrossed1}
\end{equation}
but it is nontrivial to control how the constant in the error term depends on $\tau,
\ell$, and $z$.

In the next section we will present a rigorous argument which resolves all
these caveats:
\begin{itemize}
	\item We will find the range of $z, \bar{z}$ such that Eq.\,{\eqref{unitdom}}
	is true, i.e.\,the precise formulation of DLC limit, called
	DLC$_{\ast}$ below. It will be more subtle than the condition $z \ll 1 - \bar{z}$
	previously used by many authors.
	
	\item We will show that crossed-channel blocks with $\tau$ not accumulating
	to $2 \Delta_{\phi}$, even infinitely many of them, cannot reproduce
	{\eqref{unitdom}}. Namely, the total crossed channel contribution of all
	blocks with $\lvert \tau - 2 \Delta_{\phi} \rvert > \varepsilon$ goes to zero in
	the DLC$_{\ast}$ limit, for an arbitrary $\varepsilon > 0$.
	
	\item As a part of the proof, we will show that the error term in
	{\eqref{CBcrossed1}} cannot overwhelm the leading term, in a sense which will be made precise, in an arbitrary $d \geqslant 2$.
\end{itemize}
\subsection{Rigorous proof}\label{rig-proof}

Here is a rigorous formulation of the twist accumulation result. 
\begin{theorem}
	\label{theorem:DLCtwist}Let $\phi$ be a Hermitian scalar quasiprimary in
	a unitary CFT in $d \geqslant 2$. Suppose that
	\begin{itemize}
		\item all nontrivial quasiprimary operators in the OPE $\phi \times
		\phi$ satisfy the twist gap assumption {\eqref{TGA}} with $\tau_{\rm gap}>0$;
		
		\item the spectrum of operators in this OPE is discrete, i.e.\,there is a
		finite number of operators with $\Delta \leqslant \Delta_{\ast}$ for any
		$\Delta_{\ast}$.
	\end{itemize}
	Then:
	\begin{enumerate}
	\item[(a)] $\tau_{\rm gap} \leqslant 2 \Delta_{\phi}$;
	\item[(b)] The OPE $\phi \times \phi$ contains an infinite sequence of
	quasiprimary operators $(\mathcal{O}_i)_{i = 1}^{\infty}$ whose twist $\tau_i
	\rightarrow 2 \Delta_{\phi}$ while spin $\ell_i \rightarrow \infty$.
	\end{enumerate}
\end{theorem}

The rest of this subsection will be devoted to the rigorous proof of this
theorem. Although part (b) is stronger than part (a), it
will be convenient to establish part (a) first.
	
As above, consider the 4pt function $g (z, \bar{z})$, which satisfies
the crossing equation {\eqref{crosseq}} and has expansion {\eqref{CBexp}} in
conformal blocks, each of which can be expanded in turn as in
{\eqref{CBindexp}}. We will work for $0 < z, \bar{z} < 1$. Expansions
{\eqref{CBexp}} and {\eqref{CBindexp}} converge in this range. The proof
naturally splits in several steps:
\begin{enumerate}
	\item \textit{Direct channel} (Section \ref{direct}). We will show that
	the 4pt function in the direct channel is indeed dominated by the unit
	operator, provided that one considers a modified DLC limit
	denoted $\text{DLC}_{\ast}$ and defined by
	\begin{equation}
		\text{DLC${}_{\ast}$ limit} : \qquad z \rightarrow 0, \quad \bar{z}
		\rightarrow 1, \quad \frac{z^{\tau_{\rm gap} / 2}}{(1 -
			\bar{z})^{\Delta_{\phi}}} \rightarrow 0 . \label{DLC*}
	\end{equation}
	The third requirement is what makes $\text{DLC}_{\ast}$ different from DLC.
	
	\item \textit{Crossed channel: high twists} (Section \ref{crossed-high}).
	Here we turn our attention to the crossed channel, and show that, for any
	$\varepsilon > 0$, the total contribution of operators with twists $\tau
	\geqslant 2 \Delta_{\phi} + \varepsilon$ vanishes in the
	$\text{DLC}_{\ast}$ limit.
	
	\item \textit{Crossed channel: low twists }(Section \ref{crossed-low}).
	We then consider low twists and show that the same is true for them: for any
	$\varepsilon > 0$, the total contribution of operators with $\tau \leqslant
	2 \Delta_{\phi} - \varepsilon$ also vanishes in the $\text{DLC}_{\ast}$
	limit. Here a key role will be played by the bounds on conformal blocks
	discussed in Section \ref{lemma}, which make precise Eq.
	{\eqref{CBcrossed}}.
\end{enumerate}
The proof is then completed in Section \ref{crossed-around}.

\subsubsection{Direct channel}\label{direct}

Let us carry out this plan. Since $z \rightarrow 0$ and $\bar{z} \rightarrow
1$ we may assume that $\bar{z} \geqslant z$. We choose and fix a number $z_0
\in (0, 1)$, and will assume that
\begin{equation}
	0 \leqslant z \leqslant z_0 \leqslant \bar{z} < 1 \label{CFT:regime} .
\end{equation}
The precise value of $z_0$ does not matter; we could have taken $z_0 = 1 / 2$.
In Section \ref{crossed-low} below we will have to restrict the range \label{CFT:regime}
further a bit, so that $\bar{z} / z \geqslant 1 + \delta$ for a fixed positive
$\delta$.

In this subsection we will show, using the direct channel conformal block
decomposition, that there is a constant $C$ such that
\begin{equation}
	1 \leqslant g (z, \bar{z}) \leqslant 1 + C \frac{z^{\tau_{\rm gap} / 2}}{(1 -
		\bar{z})^{\Delta_{\phi}}}  \qquad (0 \leqslant z \leqslant z_0 \leqslant
	\bar{z} < 1) . \label{CFT:LHS}
\end{equation}
The lower bound here is obvious, since the blocks are non-negative for $0
\leqslant z, \bar{z} < 1$, and $p_{\tau, \ell} \geqslant 0$. This constant $C$
will in general depend on the 4pt function we are considering, and on $z_0$.
The point is that it is independent of $z, \bar{z}$. Eq.\,{\eqref{CFT:LHS}}
implies that
\begin{equation}
	g (z, \bar{z}) \rightarrow 1 \qquad \text{as}\quad \frac{z^{\tau_{\rm gap} / 2}}{(1
		- \bar{z})^{\Delta_{\phi}}} \rightarrow 0,
\end{equation}
and in particular in the DLC$_{\ast}$ limit {\eqref{DLC*}}. Unlike the intuitive argument of
Section \ref{intuitive}, our proof of {\eqref{CFT:LHS}} will apply to the full
sum of blocks and not just to the individual blocks.

We will start with two simple lemmas.

\begin{lemma}
	\label{lemma:CBineq}For each $(\Delta, \ell)$ that satisfies the unitarity
	bound, we have
	\begin{equation}
		g_{\tau, \ell} (z, \bar{z}) \leqslant \left( {z}/{z_0} \right)^{\tau /
			2} g_{\tau, \ell} (z_0, \bar{z}) \qquad \left(
		\forall \ 0 \leqslant z \leqslant z_0 < 1, 0 \leqslant \bar{z}
		< 1 \right) \label{L2stat} .
	\end{equation}
\end{lemma}

\begin{proof}
	(See {\cite{Fitzpatrick:2012yx}}, Eq.\,(80)) This is an elementary
	consequence of Eq.\,{\eqref{CBindexp}}:
	\begin{equation}
		g_{\tau, \ell} (z, \bar{z}) \leqslant (z \bar{z})^{\tau / 2} \sum_{n,
			m \geqslant 0} a_{n, m} z_0^n \bar{z}^m = \left( {z}/{z_0}
		\right)^{\tau / 2} g_{\tau, \ell} (z_0, \bar{z}),
	\end{equation}
	where we used that $a_{n, m} \geqslant 0$ when the unitarity bound holds.
\end{proof}

By the convergence of the expansion {\eqref{CBexp}}, the 4pt function $g (z,
\bar{z})$ is bounded away from $z = 1$ and $\bar{z} = 1$. The second lemma
bounds the growth of $g (z, \bar{z})$ as one of the arguments tends to 1.

\begin{lemma}
	\label{lemma:gbound1}There exists a constant $B$ such that
	\begin{equation}
		g (z_0, \bar{z}) \leqslant \frac{B}{(1 - \bar{z})^{\Delta_{\phi}}}
		\qquad \left( \forall \ 0 \leqslant
		\bar{z} < 1 \right) . \label{gbound1}
	\end{equation}
\end{lemma}

\begin{proof}
	The constant $B$ in {\eqref{gbound1}} in general depends on the considered
	4pt function, as it will be clear from the proof, see Remark \ref{leitmotif} below.
	
	For $\bar{z} \leqslant z_0$, Eq.~{\eqref{gbound1}} is trivially true since
	the 4pt function is then bounded. 
	
	For $\bar{z} \geqslant z_0$ we use the
	crossing symmetry Eq.\,{\eqref{crosseq}} to write
	\begin{equation}
		g (z_0, \bar{z}) = \frac{1}{(1 - \bar{z})^{\Delta_{\phi}}} \times \left[
		\left( \frac{z_0 \bar{z}}{1 - z_0} \right)^{\Delta_{\phi}} g (1 - z_0, 1
		- \bar{z}) \right] . \label{crossing:z0zb}
	\end{equation}
	Since $\bar{z} \geqslant z_0$, both arguments of the 4pt function in the
	r.h.s.~are separated from 1. Hence the expression in the square brackets is
	uniformly bounded, and we obtain {\eqref{gbound1}}.
\end{proof}

\begin{remark}\label{leitmotif} Going through the proof, we see that we can take 
	\begin{equation}
		B = \max(B_1,B_2),\quad B_1= g(z_0,z_0),\ B_2=\left(\frac{z_0}{1-z_0}\right)^{\Delta_\phi} g(1-z_0,1-z_0)\,. \label{modeleq}
		\end{equation}
	This equation shows one of the main ideas of this paper, that constants in bounds of infinite sums should be expressed in terms of the 4pt function at some fixed positions. Since the 4pt function is itself an infinite sum, and it's finite by assumption, this provides us a rigorous handle to estimate various other sums. Below we will not go into trouble of writing out all our constants as explicitly as in Eq.~\eqref{modeleq}, but it can be done if needed.
\end{remark}

For any $\tau_{\ast}$, define $g_{\tau \geqslant \tau_{\ast}}$ to be the sum
of all terms in the direct channel expansion {\eqref{CBexp}} having twist
$\tau \geqslant \tau_{\ast}$:\footnote{Since the convergent expansion
	{\eqref{CBexp}} consists of positive terms, the series defining $g_{\tau
		\geqslant \tau_{\ast}}$ is convergent and $g_{\tau \geqslant \tau_{\ast}}$ is
	well defined.}
\begin{equation}
	g_{\tau \geqslant \tau_{\ast}} (z, \bar{z}) = \sum_{\tau \geqslant
		\tau_{\ast}}\sum_{\ell} p_{\tau, \ell}\, g_{\tau, \ell} (z, \bar{z}) .
	\label{ghigh}
\end{equation}
\begin{proposition}
	\label{lemmahigh}There is a constant $C$ such that
	\begin{equation}
		g_{\tau \geqslant \tau_{\ast}} (z, \bar{z}) \leqslant C
		\frac{z^{\tau_{\ast} / 2}}{(1 - \bar{z})^{\Delta_{\phi}}}  \qquad (0
		\leqslant z \leqslant z_0 \leqslant \bar{z} < 1) \label{ghighbound} .
	\end{equation}
\end{proposition}

\begin{proof}
	We have the following chain of estimates (see explanations below):
	\begin{eqnarray}
		g_{\tau \geqslant \tau_{\ast}} (z, \bar{z}) & \leqslant & (
		{z}/{z_0} )^{\tau_{\ast} / 2}\sum_{\tau \geqslant
			\tau_{\ast}} \sum_\ell p_{\tau, \ell} g_{\tau, \ell} (z_0, \bar{z}) 
		\label{line1}\\
		& \leqslant & ( {z}/{z_0} )^{\tau_{\ast} / 2} g (z_0,
		\bar{z}) \leqslant ( {z}/{z_0} )^{\tau_{\ast} / 2} \frac{B}{(1
			- \bar{z})^{\Delta_{\phi}}} . 
	\end{eqnarray}
In \eqref{line1}, we applied Lemma \ref{lemma:CBineq} to every term in
	{\eqref{ghigh}}, and used $(z / z_0)^{\tau / 2} \leqslant (z /
	z_0)^{\tau_{\ast} / 2}$. Then, by the direct channel expansion, we 
	bounded the sum in the r.h.s.~of {\eqref{line1}} by the whole of $g (z_0,
	\bar{z})$. Finally, we used Lemma \ref{lemma:gbound1}. Eq.\,{\eqref{CFT:LHS}}
	now follows with $C = B / z_0^{\tau_{\ast} / 2}$.
\end{proof}

Applying this proposition with $\tau_{\ast} = \tau_{\rm gap}$, we get the
promised Eq.\,{\eqref{CFT:LHS}}.

We thus resolved the first caveat of the intuitive argument from Section
\ref{intuitive}. The main idea was to relate the constant of the $z^{\tau}$
growth of $g_{\tau, \ell}$ to the value of the same conformal block at $z =
z_0$ (Lemma \ref{lemma:CBineq}). These constants can then be resummed into the
value of the 4pt function itself, which we know how to bound from crossing
(Lemma \ref{lemma:gbound1}).

\begin{remark}
	A weaker version of Eq.\,{\eqref{CFT:LHS}}, with $\frac{z^{\tau_{\rm gap} /
			2}}{(1 - \bar{z})^{\Delta_{\phi}}}$ replaced by $\frac{z^{\tau_{\rm gap} /
			2}}{(1 - \bar{z})^{2 \Delta_{\phi}}}$, was previously shown in
	{\cite{Qiao:2017xif}}, App. F. It would be interesting to know if Eq.
	{\eqref{CFT:LHS}} is optimal or if it can be further improved. 
\end{remark}

\begin{proof}[Proof of Theorem \ref{theorem:DLCtwist}(a)] %This can now be done easily, using Eq.\,{\eqref{CFT:LHS}}. 
	Consider the limit $z \rightarrow 0$,
	$\bar{z} \rightarrow 1$, $\frac{z^{\tau_{\rm gap} / 2}}{(1 -
		\bar{z})^{\Delta_{\phi}}} = \text{const}$. Eq.\,{\eqref{CFT:LHS}} shows
	that the 4pt function is bounded in this limit. On the other hand, from the crossed channel the 4pt function is at least as large as the
	unit operator contribution, which is given by
	\begin{equation}
		\left( \frac{z \bar{z}}{(1 - z) (1 - \bar{z})} \right)^{\Delta_{\phi}} \times 1
		\equiv \left( \frac{\bar{z}}{1 - z} \right)^{\Delta_{\phi}}
		\frac{z^{\tau_{\rm gap} / 2}}{(1 - \bar{z})^{\Delta_{\phi}}}\;
		z^{\Delta_{\phi} - \tau_{\rm gap} / 2} .
	\end{equation}
	If $\tau_{\rm gap} > 2 \Delta_{\phi}$, this blows up in the considered
	limit, leading to a contradiction.
\end{proof}

\subsubsection{Crossed channel: individual blocks}\label{lemma}

We will start with a two-sided bound on conformal blocks, 
a key ingredient for the analysis of crossed channel contributions at low twists. We
normalize conformal blocks by
\begin{equation}
	\underset{x \rightarrow 0^+}{\lim} x^{- \Delta} g_{\tau, \ell} (x, x) = 1.
	\label{CB:normalization}
\end{equation}
Define
\begin{equation}
	k_{\beta} (z) = z^{\frac{\beta}{2}} {}_2 F_1 \left( \tfrac{\beta}{2},
	\tfrac{\beta}{2} ; \beta ; z \right)
\end{equation}
and\footnote{\label{nuratio}For $d = 2$ we have $\nu = 0$ and the coefficient
	$\frac{(\nu)_{\ell}}{(2 \nu)_{\ell}}$ is defined as $\underset{\nu \rightarrow
		0^+}{\lim} \frac{(\nu)_{\ell}}{(2 \nu)_{\ell}} = \frac{1 + \delta_{\ell,
			0}}{2}$}
\begin{equation}
	F_{\tau, \ell} (z) \equiv \frac{(\nu)_{\ell}}{(2 \nu)_{\ell}} 4^{\tau / 2}
	k_{\tau + 2 \ell} (z) . \label{def:Ftaul}
\end{equation}
We also denote as usual $\bar{\rho} = \bar{z} /( 1 + \sqrt{1 - \bar{z}})^2$ (and analogously $\rho$ in terms of $z$).

\begin{lemma}[Approximate factorization]
	\label{AFLz}Let $a, b \in (0, 1)$, $a < b$. Let $d\geqslant 2$ and assume that $\Delta, \ell$ satisfy the
	unitarity bounds. If $\ell = 0$, assume in addition that $\Delta > \nu =
	\frac{d - 2}{2} .$ Then the conformal block $g_{\tau, \ell} (z, \bar{z})$
	satisfies the following two-sided bound
	\begin{equation}
		1 \leqslant \frac{g_{\tau, \ell} (z, \bar{z})}{\bar{\rho}^{\tau / 2}
			F_{\tau, \ell} (z)} \leqslant K_d (a, b) \left[ 1 + \frac{\theta (d
			\geqslant 3) \delta_{\ell 0} }{\Delta - \nu} \right] 
		\qquad (\forall\ 0 \leqslant \bar{z} \leqslant a <
		b \leqslant z < 1), \label{CB:sandwich}
	\end{equation}
	where $K_d (a, b)$ is a finite constant independent of $\tau, \ell$ and of
	$z, \bar{z}$ in the shown range. 
\end{lemma}

\begin{remark}
	\label{AFLrem}The proof of this lemma is a bit technical and is postponed to
	Appendix \ref{appendix:proofAFL}. Here we collect several remarks:
	\begin{itemize}
	\item[(a)] The lemma states that in the shown range of $z, \bar{z}$ the blocks are
	given, up to constants, by the factorized expression $\bar{\rho}^{\tau / 2}
	F_{\tau, \ell} (z)$. Excluding the case $\ell = 0$, $\Delta \rightarrow
	\nu$, the constants are independent of $\tau, \ell$. This fact will play a
	crucial role in the analysis of the low-twist contribution in Section
	\ref{crossed-low}.
	
	\item[(b)] For $\ell = 0$ we see the $\frac{\theta (d \geqslant 3)}{\Delta - \nu}$
	pole in the r.h.s.~of the bound. This is as expected since the scalar blocks
	are known to have such a singular behavior near the unitarity bound for $d
	\geqslant 3$. (For $d = 2$ the scalar block at $\Delta = \nu = 0$ coincides
	with the unit operator block.)
	
	Being singular, the scalar block with $\Delta = \nu$ ($d \geqslant 3$)
	cannot appear in the conformal block decomposition of the 4pt function of
	four identical scalar quasiprimaries that we are studying.
	
	The previous observation can also be understood by realizing that $\ell =
	0$, $\Delta = \nu$ corresponds to the exchanged scalar being free scalar
	field, call it $\chi$. In the free scalar theory, $\chi$ only appears in the OPE of
	nonidentical quasiprimaries such as $\chi^n \times \chi^{n - 1}$.
	
	\item[(c)] Using
	\begin{equation}
		\bar{z} / 4 \leqslant \bar{\rho} \leqslant \bar{z} \qquad (0 \leqslant
		\bar{z} < 1),
	\end{equation}
	the lemma implies a two-sided bound in terms of $\bar{z}^{\tau / 2} F_{\tau,
		\ell} (z)$. However in this case the lower-bound constant becomes
	$\tau$-dependent. Assuming that twist $\tau$ is bounded from above, $\tau
	\leqslant \tau_{\max}$, we can choose a $\tau$-independent lower constant
	$4^{- \tau_{\max} / 2}$. Below, Lemma \ref{AFLz} will be used to analyze
	contributions of operators at low twists, so the twist will indeed be
	bounded from above.
	
	\item[(d)] Note that, analogously to {\eqref{CBindexp}}, the blocks allow a
	positive expansion in $\rho, \bar{\rho}$:
	\begin{equation}
		g_{\tau, \ell} (z, \bar{z}) = (\rho \bar{\rho})^{\tau / 2} \sum_{n, m \in
			\mathbb{Z}_{\geqslant 0}} \tilde{a}_{n, m} \rho^n \bar{\rho}^m, \qquad
		\tilde{a}_{n, m} \geqslant 0.
	\end{equation}
	It can be shown that the sum of terms having $m = 0$ is precisely
	$\bar{\rho}^{\tau / 2} F_{\tau, \ell} (z)$. Thus
	\begin{equation}
		\bar{\rho}^{\tau / 2} F_{\tau, \ell} (z) \leqslant g_{\tau, \ell} (z,
		\bar{z}) \leqslant \bar{\rho}^{\tau / 2} F_{\tau, \ell} (z) +(\text{$m>0$ terms})\,.
	\end{equation}
	This is a simple way to prove the lower bound in {\eqref{CB:sandwich}}. This
	approach by itself does not give the needed upper bound. Indeed while the $m>0$ terms
    are $O (\bar{\rho}^{\tau / 2 + 1})$ for a fixed $z$, it is not immediately obvious how they depend on $z$, $\tau$, $\ell$.
	
	Our proof will follow a different strategy. We will first consider $d = 2$,
	using the explicit expressions of conformal blocks. We will then perform
	induction in $d$ using the dimensional reduction formula of Hogervorst
	{\cite{Hogervorst:2016hal}}.
	
	\item[(e)] Lemma \ref{AFLz} implies the well-known result that individual conformal
	blocks (except the unit block) grow logarithmically as $z \rightarrow 1$ for
	a fixed $\bar{z}$.\footnote{Surprisingly, we do not know any other general rigorous
		proof of this result.} Indeed we have
	\begin{equation}
		F_{\tau, \ell} (z) \sim \frac{(\nu)_{\ell}}{(2 \nu)_{\ell}} 4^{\tau / 2}
		C_{\tau + 2 \ell} \log \frac{1}{1 - z} \qquad (z \rightarrow 1), \qquad
		C_{\beta} = \frac{\Gamma (1 / 2 + \beta / 2) 2^{\beta - 1}}{\sqrt{\pi}
			\Gamma (\beta / 2)} .
	\end{equation}
\end{itemize}
\end{remark}

\begin{lemma}
	\label{ind}Each individual conformal block contribution in the crossed
	channel vanishes in the DLC$_{\ast}$ limit.
\end{lemma}

\begin{proof}
	First consider the unit operator block, for which we have, keeping just potentially singular factors,
	and using $\tau_{\rm gap} \leqslant 2 \Delta_{\phi}$ (as we know from Part
	(a)):
	\begin{equation}
		\left( \frac{z}{1 - \bar{z}} \right)^{\Delta_{\phi}} \leqslant
		\frac{z^{\tau_{\rm gap} / 2}}{(1 - \bar{z})^{\Delta_{\phi}}},
	\end{equation}
	which vanishes in the DLC$_{\ast}$ limit.
	
	Next consider any non-unit block $g_{\tau, \ell}$, $\tau \geqslant
	\tau_{\rm gap} > 0$. By Remark \ref{AFLrem}(c,e), Lemma \ref{AFLz}
	implies a logarithmic bound:
	\begin{eqnarray}
		g_{\tau, \ell} (1 - z, 1 - \bar{z}) & \leqslant & C (1 - \bar{z})^{\tau /
			2} \log \frac{1}{z}  \label{log-bound1}\leqslant C (1 - \bar{z})^{\tau_{\rm gap} / 2} \log \frac{1}{z} . 
	\end{eqnarray}
	Thus, the crossed channel contribution of $g_{\tau, \ell}$ is bounded by a
	constant times
	\begin{equation}
		\left( \frac{z}{1 - \bar{z}} \right)^{\Delta_{\phi}} (1 -
		\bar{z})^{\tau_{\rm gap} / 2} \log \frac{1}{z} . \label{discrbound}
	\end{equation}
	Denoting $\kappa = 1 - \frac{\tau_{\rm gap}}{2 \Delta_{\phi}}$, this can be
	rewritten as
	\begin{equation}
		\left[ \frac{z^{\tau_{\rm gap} / 2}}{(1 - \bar{z})^{\Delta_{\phi}}}
		\right]^{\kappa} \times z^{\Delta_{\phi} - \tau_{\rm gap} \kappa / 2} \log
		\frac{1}{z} .
	\end{equation}
Since $\kappa \geqslant 0$ (part (a) of the theorem), the first
	factor is bounded in the DLC$_{\ast}$ limit. At the same time the second
	factor goes to zero, since a positive power of $z$ kills the log.
\end{proof}

\subsubsection{Crossed channel: high twists}\label{crossed-high}

In the previous subsection we studied the contribution of individual conformal
blocks in the crossed channel. In this subsection we will study the
\textit{total} crossed channel contribution of operators with twists $\tau
\geqslant 2 \Delta_{\phi} + \varepsilon$ to the r.h.s.~of the crossing
equation, and we will show that it vanishes in the $\text{DLC}_{\ast}$ limit.
The argument for this is rather simple and robust, and does not need
considerations of the previous subsection.

We need to estimate $g_{\tau \geqslant \tau_{\ast}} (1 - z, 1 - \bar{z})$ for
$\tau_{\ast} = 2 \Delta_{\phi} + \varepsilon$. The needed estimate follows
from Proposition \ref{lemmahigh}, which we need to apply exchanging $z$ with
$\bar{z}$:\footnote{The proposition remains
	true in this formulation, as it is clear from its proof.}
\begin{equation}
	g_{\tau \geqslant \tau_{\ast}} (z, \bar{z}) \leqslant C
	\frac{\bar{z}^{\tau_{\ast} / 2}}{(1 - z)^{\Delta_{\phi}}}
	\qquad (0 < \bar{z} \leqslant z_0 \leqslant z < 1) .
\end{equation}
Replacing $z \rightarrow 1 - z$, $\bar{z} \rightarrow 1 - \bar{z}$, $z_0
\rightarrow 1 - z_0$, we obtain:
\begin{equation}
	g_{\tau \geqslant \tau_{\ast}} (1 - z, 1 - \bar{z}) \leqslant C \frac{(1 -
		\bar{z})^{\tau_{\ast} / 2}}{z^{\Delta_{\phi}}} \qquad (0 < z \leqslant z_0
	\leqslant \bar{z} < 1) . \label{high-twist-bound}
\end{equation}
This bound has a powerlaw growth as $z \rightarrow 0$ compared to the bound
{\eqref{log-bound1}} on the individual blocks which only grows
logarithmically. The point is that it applies to the sum of all high-twist
blocks. Thanks to the $(1 - \bar{z})^{\tau_{\ast} / 2}$ factor, we will still
be able to show that the contribution to the crossing equation vanishes.

Plugging bound {\eqref{high-twist-bound}} into the r.h.s.~of the crossing
equation, the singular factors $z^{\Delta_{\phi}}$ cancel and we get:
\begin{equation}
	\left( \frac{z \bar{z}}{(1 - z) (1 - \bar{z})} \right)^{\Delta_{\phi}}
	g_{\tau \geqslant \tau_{\ast}} (1 - z, 1 - \bar{z}) \leqslant \text{const} .
	(1 - \bar{z})^{\tau_{\ast} / 2 - \Delta_{\phi}} = \text{const} . (1 -
	\bar{z})^{\varepsilon / 2}, \label{end-high}
\end{equation}
where we used $\tau_{\ast} = 2 \Delta_{\phi} + \varepsilon$.

As promised, we see that the crossed channel contribution of twists $\tau
\geqslant 2 \Delta_{\phi} + \varepsilon$ vanishes in the
$\text{DLC}_{\ast}$ limit. (We only need
$\bar{z} \rightarrow 1$ but not the full strength of $\text{DLC}_{\ast}$.)

\begin{remark}
In $d \geqslant
3$ the conformal blocks are symmetric w.r.t.~$z \leftrightarrow \bar{z}$. In $d = 2$ the transformation $z \leftrightarrow \bar{z}$ is a parity
	transformation. If a 2D CFT does not preserve parity, the blocks are not symmetric. They are characterized by a pair $(h,
	\bar{h})$ with $h + \bar{h} = \Delta$ and $h - \bar{h}$ an integer. Without the parity assumption, the 2D CFT
	spectrum is not necessarily symmetric w.r.t. $h \leftrightarrow \bar{h}$. 
	
	The above proof works in 2D independently of the parity assumption. This is because the $(h,
	\bar{h})$ and $(\bar{h},h)$ parts of parity-invariant 2D conformal blocks separately satisfy the assumptions which went into the proof of Proposition \ref{lemmahigh}.
	
    In the 2D case without parity, the theorem can be strengthened:
	one can show that there is an infinite series of blocks with $h \rightarrow
	\Delta_{\phi}$ and $\bar{h} \rightarrow \infty$, and vice versa. To keep
	the exposition simple we will not consider this more general 2D case, and we
	will continue assuming that the conformal blocks are symmetric. However in
	Section \ref{LMB}, when we talk about the modular bootstrap, we consider the
	general case.
\end{remark}

\subsubsection{Crossed channel: low twists}\label{crossed-low}

We saw in the previous subsection that the high-twist contribution is
suppressed at small $1 - \bar{z}$ simply thanks to the $(1 - \bar{z})^{\tau /
	2}$ factor which wins over the $(1 - \bar{z})^{\Delta_{\phi}}$ denominator
in the crossing relation, see Eq.\,{\eqref{end-high}}. For low twists, the argument will be more subtle. We will have to use that the low twist
conformal blocks have ``wrong scaling with $(1 - \bar{z})$'' to cancel the $(1
- \bar{z})^{\Delta_{\phi}}$ denominator (recall the naive argument). This
will rely on Lemma \ref{AFLz}.

Let us pick $\tau_{\ast} = 2 \Delta_{\phi} - \varepsilon$, $\varepsilon >
0$. Define $g_{\tau \leqslant \tau_{\ast}} (z, \bar{z})$ by
\begin{equation}
	g_{\tau \leqslant \tau_{\ast}} (z, \bar{z}) = \sum_{\tau \leqslant
		\tau_{\ast}}\sum_{\ell} p_{\tau, \ell}\, g_{\tau, \ell} (z, \bar{z})\, . \label{glow}
\end{equation}
This is analogous to Eq.\,{\eqref{ghigh}} for high twists. Using Lemma
\ref{AFLz}, we can show that every block in $g_{\tau \leqslant \tau_{\ast}}$
satisfies a two-sided bound:
\begin{equation}
	C_1 \leqslant \frac{g_{\tau, \ell} (z, \bar{z})}{ \bar{z}^{\tau / 2}
		F_{\tau, \ell} (z)} \leqslant C_2 \qquad (0 < \bar{z} \leqslant a < b
	\leqslant z < 1), \label{2sidebnd}
\end{equation}
where $C_1, C_2$ are independent of $\tau$ and $\ell$. [The precise values of
$a, b$ won't matter, we could have chosen e.g. $a = 1 / 3$, $b = 2 / 3$. Below $a$ and $b$ are kept fixed.]
Indeed:
\begin{itemize}
	\item The unit block obviously satisfies such a bound.
	
	\item By Remark \ref{AFLrem}(e), all blocks satisfy a uniform lower bound of
	this form, since we are restricting $\tau \leqslant \tau_{\ast}$.
	
	\item By Lemma \ref{AFLz}, blocks of spin $\ell \geqslant 1$ satisfy an
	upper bound uniform in $\tau, \ell$. By the unitarity bound $\Delta
	\geqslant \ell + 2 \nu$ these blocks stay away from the $\frac{1}{\Delta -
		\nu}$ pole in the r.h.s.~of {\eqref{CB:sandwich}}.
	
	\item For the upper bound on the scalar blocks, we argue as follows. The $d
	= 2$ case is clear since the $\frac{1}{\Delta - \nu}$ pole is absent. For $d
	= 3$ we still get a uniform bound with $\frac{1}{\Delta - \nu} \rightarrow
	\frac{1}{\Delta_{\min}^0 - \nu}$,  where $\Delta_{\min}^0$ is the lowest
	scalar quasiprimary dimension in the conformal block decomposition of $\langle
	\phi \phi \phi \phi \rangle$. We are using here the assumption
	that the spectrum is discrete. Note that $\Delta_{\min}^0 > \nu$, since the
	$\Delta = \nu$ scalar block does not appear for $d \geqslant 3$, see Remark
	\ref{AFLrem}(b).
\end{itemize}
Using Eq.\,{\eqref{2sidebnd}}, we have the following bound for every block in
$g_{\tau \leqslant \tau_{\ast}}$:
\begin{equation}
	g_{\tau, \ell} (z, \bar{z}) \leqslant C \left(
	\frac{\bar{z}}{\bar{z}'} \right)^{\tau / 2} g_{\tau, \ell} (z,
	\bar{z}') \qquad \forall\ 0 < \bar{z}, \bar{z}' \leqslant a < b
	\leqslant z < 1,
\end{equation}
where $C = C_2 / C_1$ independent of $\tau$ and $\ell$. Let us assume
further that $\bar{z} / \bar{z}' \geqslant 1$. The precise value of $\bar{z}'$
will be chosen below. Then, using $\tau \leqslant \tau_{\ast}$, we can
continue the estimate as
\begin{equation}
	g_{\tau, \ell} (z, \bar{z}) \leqslant C \left(
	\frac{\bar{z}}{\bar{z}'} \right)^{\tau_{\ast} / 2} g_{\tau, \ell} (z,
	\bar{z}') \qquad \forall\ 0 < \bar{z}' \leqslant \bar{z} \leqslant a
	< b \leqslant z < 1 \quad (\tau \leqslant \tau_{\ast}) . \label{nontriv}
\end{equation}
Summing this equation over all blocks, we get
\begin{equation}
	\begin{split}
		g_{\tau \leqslant \tau_{\ast}} (z, \bar{z}) & \leqslant C \left(
		\frac{\bar{z}}{\bar{z}'} \right)^{\tau_{\ast} / 2} g_{\tau \leqslant
			\tau_{\ast}} (z, \bar{z}')
		\leqslant C \left( \frac{\bar{z}}{\bar{z}'} \right)^{\tau_{\ast} /
			2} g (z, \bar{z}'),  \\
	\end{split}
\end{equation}
where we also used that $g_{\tau \leqslant \tau_{\ast}}$ is bounded by the
full 4pt function. Moving to the crossed channel variables $z \rightarrow 1 -
z$, $\bar{z} \rightarrow 1 - \bar{z}$, we obtain:
\begin{equation}
	\begin{split}
		g_{\tau \leqslant \tau_{\ast}} (1 - z, 1 - \bar{z}) \leqslant\,& C\left(
		\frac{1 - \bar{z}}{1 - \bar{z}'} \right)^{\tau_{\ast} / 2} g(1 - z, 1 -
		\bar{z}') \hspace{3em} \\
		\forall\ 0 < z \leqslant 1 - b <& 1 - a
		\leqslant \bar{z} \leqslant \bar{z}' < 1 . \\
	\end{split}
\end{equation}
Plugging this into the r.h.s.~of the crossing equation, we have, for the same $z,\bar{z},\bar{z}'$,
	\begin{align}
		\left( \frac{z \bar{z}}{(1 - z) (1 - \bar{z})} \right)^{\Delta_{\phi}}
		g_{\tau \leqslant \tau_{\ast}} (1 - z, 1 - \bar{z}) & \leqslant C
		\left( \frac{1 - \bar{z}}{1 - \bar{z}'} \right)^{\tau_{\ast} / 2} \left(
		\frac{z \bar{z}}{(1 - z) (1 - \bar{z})} \right)^{\Delta_{\phi}} g (1 - z,
		1 - \bar{z}') \nonumber\\
		& = C \left( \frac{1 - \bar{z}}{1 - \bar{z}'} \right)^{\tau_{\ast} /
			2} \left( \frac{\bar{z}}{\bar{z}'} \right)^{\Delta_{\phi}} \left( \frac{1 -
			\bar{z}'}{1 - \bar{z}} \right)^{\Delta_{\phi}} g (z, \bar{z}') \nonumber\\
		& \leqslant C \left( \frac{1 - \bar{z}'}{1 - \bar{z}}
		\right)^{\varepsilon / 2} g (z, \bar{z}'), \label{lowtwistend}
	\end{align}
using the crossing equation for $g (z, \bar{z}')$ in the
second line, and $\left( \frac{\bar{z}}{\bar{z}'} \right)^{\Delta_{\phi}}
\leqslant 1$ in the third line.

Now we choose
\begin{equation}
	\bar{z}' = 1 - z^{\tau_{\rm gap} / (2 \Delta_{\phi})} .
\end{equation}
By Eq.\,{\eqref{CFT:LHS}} we know that $g (z, \bar{z}')$ is bounded as $z
\rightarrow 0$. On the other hand
\begin{equation}
	\frac{1 - \bar{z}'}{1 - \bar{z}} = \frac{z^{\tau_{\rm gap} / (2
			\Delta_{\phi})}}{1 - \bar{z}} \rightarrow 0 \qquad \text{in the DLC$_{\ast}$ limit} .
\end{equation}
Thus, the r.h.s.~of \eqref{lowtwistend} goes to zero, and we conclude that the crossed-channel contribution of low twists vanishes in the
DLC$_*$ limit.

The main idea of the low twist analysis is illustrated in Fig. \ref{fig-low}.

\begin{figure}[h]
\centering
	\resizebox{190pt}{116pt}{\includegraphics{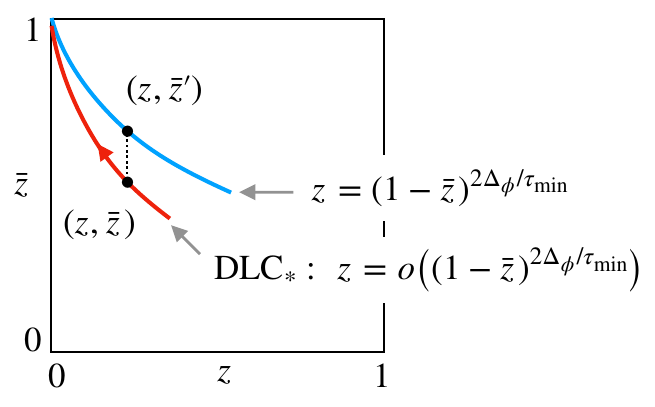}}
	\caption{\label{fig-low}To show that the low twists vanish in the
		DLC$_{\ast}$ limit (red curve) we show that they are (a) bounded on the blue
		curve $z = (1 - \bar{z})^{2 \Delta_{\phi} / \tau_{\rm gap}}$ and (b) get
		suppressed when moving from the blue to the red curve.}
\end{figure}

\subsubsection{End of proof of Theorem
	\ref{theorem:DLCtwist}}\label{crossed-around}

So far we have shown that, in the DLC$_{\ast}$ limit,
\begin{itemize}
	\item the l.h.s.~of the crossing equation (direct channel) goes to 1
	
	\item in the r.h.s.~of the crossing equation, the total contribution of
	twists $\tau$ with $\lvert\tau - 2 \Delta_{\phi}\rvert \geqslant \varepsilon$
	vanishes, for any $\varepsilon > 0$.
\end{itemize}
Putting these two things together, we conclude that the total contribution of
twists $\tau \in (2 \Delta_{\phi} - \varepsilon, 2 \Delta_{\phi} +
\varepsilon)$ goes to 1. Combining this with Lemma \ref{ind}, we conclude that
there is an infinite number of operators of twists $\tau \in (2
\Delta_{\phi} - \varepsilon, 2 \Delta_{\phi} + \varepsilon)$ for any
$\varepsilon > 0$. Hence there is a sequence of operators $\mathcal{O}_i$
whose twists $\tau_i \rightarrow 2 \Delta_{\phi}$. Since we are assuming
that the spectrum is discrete, it follows that their spins $\ell_i \rightarrow
\infty$. Theorem \ref{theorem:DLCtwist}(b) is proved.

\begin{remark}
	The nontrivial equation was {\eqref{nontriv}}. Compare this
	equation to Lemma \ref{lemma:CBineq}. In that lemma, we bound $g_{\tau,
		\ell} (z, \bar{z})$ by $g_{\tau, \ell} (z_0, \bar{z})$ at a
	\textit{larger} value of the argument: $z_0 \geqslant z$. That was
	trivial. In contrast, Eq.\,{\eqref{nontriv}} bounds $g_{\tau, \ell} (z,
	\bar{z})$ by $g_{\tau, \ell} (z, \bar{z}')$ at a
	\textit{smaller} value of the argument: $\bar{z}' \leqslant \bar{z}$. This
	is nontrivial and is only possible because of the two-sided bound
	{\eqref{2sidebnd}} implied by Lemma \ref{AFLz}.
	\end{remark}

\begin{remark} \label{collinear}
	It is instructive to see what happens when we expand the 4pt function not in
	full conformal blocks but in collinear blocks, which are given by
	$\bar{z}^{\tau / 2} F_{\tau, \ell} (z)$.\footnote{Collinear blocks are
		blocks associated with collinear primaries, which are annihilated by the
		generator $\bar{L}_1$ of the 2D global conformal group acting in the $z,
		\bar{z}$ plane, but not necessarily by the generators $L_1$ and $K_i$. Every
		conformal quasiprimary multiplet splits into infinitely many collinear primary
		multiplets.} Collinear blocks scale as an exact power in $\bar{z}$, so
	bound {\eqref{nontriv}} for them becomes trivial. We needed to work with
	full conformal blocks since we wanted to prove a statement about the
	existence of \textit{conformal} quasiprimaries with $\tau \rightarrow 2
	\Delta_{\phi}$. If we work with collinear blocks, the proof simplifies
	since we don't need Lemma \ref{AFLz}, but we only get a weaker result about
	the existence of \textit{collinear} primaries with $\tau \rightarrow 2
	\Delta_{\phi}$. One could start with such a weaker result and try to
	argue that it implies the stronger result about conformal primaries (the
	strategy tried in {\cite{Komargodski:2012ek}}). This is automatic if $2
	\Delta_{\phi}$ is close to the unitarity bound, but in general it does
	not appear easy to argue this rigorously, without effectively using Lemma
	\ref{AFLz}.
\end{remark}

\subsection{Generalizations}

\subsubsection{Continuous spectrum case}

In the statement of Theorem \ref{theorem:DLCtwist}, we assumed that the
spectrum in the $\phi \times \phi$ OPE is discrete. In this section we
would like to show that the theorem is actually true even if the spectrum is
continuous. This generalization could be interesting for unitary CFTs with a
continuous spectrum satisfying a twist gap assumptions, although we don't know
any examples of such theories. In the continuum spectrum case, Eq.
{\eqref{CBexp}} should be replaced by

\begin{equation}
	g (z, \bar{z}) = 1 + \sum_{\ell} \int d \mu_{\ell} (\tau) g_{\tau, \ell} (z,
	\bar{z}), \label{CBexp-cont}
\end{equation}
where $d \mu_{\ell}$ are non-negative measures having support at $\tau
\geqslant \tau_{\rm gap}$.

In the proof of Theorem \ref{theorem:DLCtwist}, the discrete spectrum
assumption was used in only two places:

D1. In Section \ref{crossed-low}, when we used Lemma \ref{AFLz} to estimate
the contribution of low-twist operators in the crossed channel, the discrete
spectrum gave us a lower bound $\Delta \geqslant \Delta_{\min}^0 > \nu$ in the
$\ell = 0$ case. This made the upper bound of Lemma \ref{AFLz} uniform in
$\tau$ and $\ell$.

D2. In Section \ref{crossed-around}, when concluding the proof of the theorem.
Here we used that any sequence of operators whose twist converges to a limit
has spin going to infinity.

In the continuum spectrum case, we have to argue differently. We have the following
key result which generalizes Lemma \ref{ind}:

\begin{proposition}
	\label{prop:contsubleading}For any finite $\tau_{\max}$ and $\ell_{\max}$,
	the total crossed-channel contribution of blocks with $\tau \leqslant
	\tau_{\max}, \ell \leqslant \ell_{\max}$ vanishes in the $\text{DLC}_{\ast}$
	limit. This is true for both continuum and discrete spectrum.
\end{proposition}

	The proof of this proposition will use the following uniform logarithmic bound
	on the conformal blocks, which strengthens Remark
	\ref{AFLrem}(e).
	
	\begin{lemma}
		\label{lemma:improvedlog}For any $a, b \in (0, 1)$, $a < b$ and finite
		$\tau_{\max}$, $\ell_{\max}$, there exist a finite number $B$ (which only
		depends on $a$, $b$, $\tau_{\max}$ and $\ell_{\max}$) such that
		\begin{equation}
			\frac{g_{\tau, \ell} (z, \bar{z})}{g_{\tau, \ell} (b, a)} \leqslant B
			\bar{z}^{\tau / 2} \log \left( \frac{1}{1 - z} \right)
			\qquad \forall\ 0 \leqslant \bar{z} \leqslant a <
			b \leqslant z < 1,  \label{improvedbound}
		\end{equation}
		for any $\tau \leqslant
		\tau_{\max}$, $\ell \leqslant \ell_{\max}$ such that $(\tau, \ell)$ satisfies the unitarity bounds, i.e.\,$\tau \geqslant \nu$ for $\ell = 0$ and $\tau \geqslant 2
		\nu$ for $\ell \geqslant 1$.
	\end{lemma}
	
	The proof of this lemma is given in Appendix \ref{app:improvedlog}. Here we
	only remark on the relation of Lemma \ref{lemma:improvedlog} to Lemma
	\ref{AFLz}. For $\ell \geqslant 1$ the bound (\ref{improvedbound}) follows
	from Lemma \ref{AFLz}, using the upper bound for $g_{\tau, \ell} (z,
	\bar{z})$, the lower bound for $g_{\tau, \ell} (b, a)$ and the logarithmic
	bound of the hypergeometric function. The $\ell = 0$ case is more subtle: it
	does not follow from Lemma \ref{AFLz} because of the $(\Delta - \nu)^{- 1}$
	singularity appearing in the upper bound of (\ref{CB:sandwich}). In Appendix
	\ref{app:improvedlog}, we will argue that the lower bound on $g_{\Delta, 0}
	(b, a)$ also has the $(\Delta - \nu)^{- 1}$ singularity, and hence the
	upper bound on the ratio $g_{\Delta, 0} (z, \bar{z}) / g_{\Delta, 0} (b, a)$
	does not have this singularity.
	
\begin{proof}[Proof of Proposition
	\ref{prop:contsubleading}] We already know by Lemma \ref{ind} that the
	unit operator term in the crossed channel vanishes in the DLC$_{\ast}$ limit, so we
	will not include it in the argument below. Pick for definiteness $a =
	\frac{1}{3}, b = \frac{2}{3}$. In the DLC$_{\ast}$ limit, eventually we will
	have
	\begin{equation}
		0 \leqslant 1 - \bar{z} \leqslant {\textstyle \frac{1}{3}} < {\textstyle \frac{2}{3} }\leqslant 1 -
		z < 1,
	\end{equation}
	so Lemma \ref{lemma:improvedlog} applies with the substitution $z
	\rightarrow 1 - z, \bar{z} \rightarrow 1 - \bar{z}$. Then in the crossed
	channel, the total contribution of operators with $\tau \leqslant \tau_{\max}, \ell
	\leqslant \ell_{\max}$ is bounded as follows:
		\begin{align}
			\sum_{\ell \leqslant \ell_{\max}} \int_{\tau_{\rm gap}}^{\tau_{\max}} d
			\mu_{\ell} (\tau) g_{\tau, \ell} (1 - z, 1 - \bar{z}) \leqslant & B \log
			\frac{1}{z} \sum_{\ell \leqslant \ell_{\max}}
			\int_{\tau_{\rm gap}}^{\tau_{\max}} d \mu_{\ell} (\tau) (1 - \bar{z})^{\tau /
				2} g_{\tau, \ell} \left({ \textstyle \frac{2}{3}, \frac{1}{3} }\right)  \nonumber\\
			\leqslant & B \log \frac{1}{z} (1 - \bar{z})^{\tau_{\rm gap} / 2}
			\sum_{\ell \leqslant \ell_{\max}} \int_{\tau_{\rm gap}}^{\tau_{\max}} d
			\mu_{\ell} (\tau) g_{\tau, \ell} \left( { \textstyle \frac{2}{3}, \frac{1}{3} } \right) \nonumber\\
			\leqslant & B \log \frac{1}{z}  (1 - \bar{z})^{\tau_{\rm gap} / 2}  
			g\bigl( { \textstyle \frac{2}{3}, \frac{1}{3} }\bigr) \nonumber\\
			&\left( 0 \leqslant z
			\leqslant { \textstyle \frac{1}{3},\ \frac{2}{3} }\leqslant \bar{z} < 1 \right). 
		\end{align}
	Here in the first line we used Lemma \ref{lemma:improvedlog}, in the second
	line we used $\tau \geqslant \tau_{\rm gap}$, and in the last line we bounded
	the partial sum of the conformal blocks by the full conformal block
	expansion of $g \left( \frac{2}{3}, \frac{1}{3} \right)$, assumed finite. Now together with the prefactor $\left( \frac{z \bar{z}}{(1 - z)
		(1 - \bar{z})} \right)^{\Delta_{\phi}}$ of the crossed channel, we have
	the bound for $0 \leqslant z \leqslant \frac{1}{3}, \frac{2}{3} \leqslant
	\bar{z} < 1$:
	\begin{equation}
		\left( \frac{z \bar{z}}{(1 - z) (1 - \bar{z})} \right)^{\Delta_{\phi}}
		\sum_{\ell \leqslant \ell_{\max}} \int_{\tau_{\rm gap}}^{\tau_{\max}} d
		\mu_{\ell} (\tau) g_{\tau, \ell} (1 - z, 1 - \bar{z}) \leqslant B'
		\frac{z^{\Delta_{\phi}}}{(1 - \bar{z})^{\Delta_{\phi} - \tau_{\rm gap}
				/ 2}} \log (1 / z), \label{contbound}
	\end{equation}
	where we used $(1 - z)^{- \Delta_{\phi}} \leqslant b^{-
		\Delta_{\phi}}$, $\bar{z}^{\Delta_{\phi}} \leqslant 1$ and we define
	$B' \assign B b^{- \Delta_{\phi}} g \left( \frac{2}{3}, \frac{1}{3}
	\right) < + \infty$. The r.h.s.~of {\eqref{contbound}} is the same quantity
	as Eq.\,{\eqref{discrbound}} in the proof of Lemma \ref{ind}, and as shown
	there it vanishes in the DLC$_{\ast}$ limit. This finishes the proof.
\end{proof}

Now we can get around the two uses D1,D2 of the discrete spectrum assumption.

For D1, the total low twist ($\tau \leqslant 2 \Delta_{\phi} -
\varepsilon$) spin-0 contribution vanishes by Proposition
\ref{prop:contsubleading}. So we only have to consider low twist and $\ell
\geqslant 1$, in which case the upper bound of Lemma \ref{AFLz} is uniform in
$\tau$ and $\ell$, without the need for the discrete spectrum assumption. \
The rest of the argument in Section \ref{crossed-low} is unchanged.

For D2, we argue as follows. We know that the total crossed-channel
contribution of twists $\tau \in (2 \Delta_{\phi} - \varepsilon, 2
\Delta_{\phi} + \varepsilon)$ goes to 1 for any $\varepsilon > 0$. On the
other hand the contribution of twists $\tau \in (2 \Delta_{\phi} -
\varepsilon, 2 \Delta_{\phi} + \varepsilon)$, limiting to spins $\ell
\leqslant \ell_{\max}$, goes to zero (Proposition \ref{prop:contsubleading}).
Thus, for any $\varepsilon > 0$ and any $\ell_{\max} < \infty$, the set of
operators with $\tau \in (2 \Delta_{\phi} - \varepsilon, 2 \Delta_{\phi}
+ \varepsilon)$ and $\ell > \ell_{\max}$ is not empty. This implies the
existence of a sequence of operators with twist $\tau_i \rightarrow 2
\Delta_{\phi}$ and spin $\ell_i \rightarrow \infty$.

\subsubsection{Non-identical scalars and spinning
	operators}\label{non-identical}

Theorem \ref{theorem:DLCtwist} should allow the following generalization to
the 4pt function $\langle \phi_1 \phi_2 \phi_1 \phi_2 \rangle$
where $\phi_1, \phi_2$ are non-identical scalar case:

\textit{Let $\phi_1$ and $\phi_2$ be two Hermitian scalar
	quasiprimaries in a unitary CFT in $d \geqslant 2$. Suppose that all operators
	in the OPEs $\phi_1 \times \phi_1$ and $\phi_2 \times \phi_2$,
	except for the unit operators, satisfy the positive twist gap assumption. Then
	the OPE $\phi_1 \times \phi_2$ contains an infinite sequence of quasiprimary
	operators $(\mathcal{O}_i)_{i = 1}^{\infty}$ whose twist $\tau_i \rightarrow
	\Delta_1 + \Delta_2$ while spin $\ell_i \rightarrow \infty$. $(\Delta_i \equiv
	\Delta_{\phi_i})$}

Such a generalization was considered already in {\cite{Komargodski:2012ek}}.
It can be argued at an intuitive level very similarly to Section
\ref{intuitive}. Here we give a sketch of how a rigorous proof could proceed.
It will be clear that the result may require a slight modification in the special case $\Delta_1 + \Delta_2 =
d$.

We consider the crossing equations of the 4pt function $\langle
\phi_1 \phi_1 \phi_2 \phi_2 \rangle$:
\begin{equation}
	g_{1122} (z, \bar{z}) = \frac{(z \bar{z})^{\Delta_1}}{[(1 - z) (1 -
		\bar{z})]^{\frac{\Delta_1 + \Delta_2}{2}}} g_{2112} (1 - z, 1 - \bar{z})
	\qquad (0 < z, \bar{z} < 1), \label{crossingnonid}
\end{equation}
where $g_{i j k l} (z, \bar{z})$ is defined by
\begin{equation}
	\langle \phi_i (x_1) \phi_j (x_2) \phi_k (x_3) \phi_l (x_4) \rangle =
	\frac{1}{(x_{12}^2)^{\frac{\Delta_i + \Delta_j}{2}}
		(x_{34}^2)^{\frac{\Delta_k + \Delta_l}{2}}} \left( \frac{x_{14}^2}{x_{24}^2}
	\right)^{\frac{\Delta_j - \Delta_i}{2}} \left( \frac{x_{14}^2}{x_{13}^2}
	\right)^{\frac{\Delta_k - \Delta_l}{2}} g_{i j k l} (z, \bar{z}) .
	\label{CFT4ptnonid}
\end{equation}
The direct channel decomposition is given by
\begin{equation}
	g_{1122} (z, \bar{z}) = 1 + \underset{\mathcal{O} \neq 1}{\sum}
	C_{11\mathcal{O}} C_{\bar{2} \bar{2} \mathcal{O}}^{\ast} g_{1122,
		\mathcal{O}} (z, \bar{z}),
\end{equation}
where the sum is over all the quasiprimaries on the OPE except for the identity
operator (which contributes the first term ``1''). It is not a positive sum,
but using the Cauchy-Schwarz inequality, and crossing equations for $\langle
\phi_1 \phi_1 \phi_1 \phi_1 \rangle$ and $\langle \phi_2
\phi_2 \phi_2 \phi_2 \rangle$, we can show that for an arbitrary
fixed $z_0 \in (0, 1)$,
\begin{equation}
	g_{1122} (z, \bar{z}) = 1 + O \left( \frac{z^{\tau_{\rm gap} / 2}}{(1 -
		\bar{z})^{\frac{\Delta_1 + \Delta_2}{2}}} \right) \qquad (0 \leqslant z \leqslant z_0 \leqslant \bar{z} < 1) .
\end{equation}
This estimate, which generalizes Eq.\,{\eqref{CFT:LHS}}, suggests that the
$\text{DLC}_{\ast}$ limit for non-identical scalars should be defined as
\begin{equation}
	\text{DLC}_{\ast} :\qquad z \rightarrow 0,
	\quad \bar{z} \rightarrow 1, \quad \frac{z^{\tau_{\rm gap}}}{(1 - \bar{z})^{\Delta_1 + \Delta_2}}
	\rightarrow 0.
\end{equation}
Then we consider the crossed channel, which is a positive sum
\begin{equation*}
	\begin{split}
		g_{2112} (z, \bar{z}) = \underset{\mathcal{O}}{\sum} \lvert C_{21\mathcal{O}}
		\rvert^2 g_{2112, \mathcal{O}} (z, \bar{z}) .
	\end{split}
\end{equation*}
Generalizing the arguments of Section \ref{crossed-high}, we can still easily
show that in the crossed channel, the total contribution from operators with
twist $\tau \geqslant \Delta_1 + \Delta_2 + \varepsilon$ vanishes in the
$\text{DLC}_{\ast}$ limit.

It remains then to deal with the low-twist part. We expect that an analogue
of Lemma \ref{AFLz} holds for the case of non-identical external scalar
quasiprimaries. However, when the external scaling dimensions are not equal, the
conformal blocks are known to have extra poles at $\tau = 2 \nu (\equiv d -
2)$ for $\ell \geqslant 1$. These poles are due to the level-1 null states
$P_{\mu_1} \vert \mathcal{O}^{\mu_1 \mu_2 \ldots \mu_{\ell}} \rangle$,
where
$\mathcal{O}$ is the exchanged spin-$\ell$ quasiprimary operator in the OPE. Based
on this reason, we believe that the $\ell \geqslant 1$ conformal block
$g^{2112}_{\tau, \ell}$ of non-identical external scalars should satisfy the
following approximate factorization property, analogous to the bound
{\eqref{2sidebnd}} in the current situation:
\begin{equation}
	C_1 \leqslant \frac{g_{\tau, \ell}^{2112} (z, \bar{z})}{\bar{z}^{\tau / 2}
		F_{\tau, \ell}^{2112} (z) + \frac{1}{\tau - 2 \nu} \bar{z}^{\tau / 2 + 1}
		F_{\tau + 2, \ell}^{2112} (z)} \leqslant C_2 \qquad
	(\forall\ 0 \leqslant \bar{z} \leqslant a < b \leqslant z < 1) .
	\label{gen-appr}
\end{equation}
Here $F_{\tau, \ell}^{2112} (z) = z^{\tau / 2 + \ell} {}_2 F_1 \left(
\frac{\tau + 2 \ell + \Delta_1 - \Delta_2}{2}, \frac{\tau + 2 \ell + \Delta_1
	- \Delta_2}{2} ; \tau + 2 \ell ; z \right)$. The constants $C_1, C_2$ are
uniform in $\tau, \ell$ as long as $\tau \leqslant \tau_{\max} < \infty$ (we
assume that $a, b, \Delta_1, \Delta_2$ are fixed). The second term in the
denominator is new with respect to the identical external dimensions.
\begin{remark}
	When $\Delta_1\neq\Delta_2$, $\tau$ is always strictly larger than $2\nu$. This is because the spin-$\ell$ operaotors with $\tau=2\nu$, i.e.\,the conserved current operators, only appears in the $\phi_1\times\phi_2$ OPE with $\Delta_1=\Delta_2$, as a consequence of Ward identity.
\end{remark}

We did not try to work out a proof of {\eqref{gen-appr}}, but we suspect that
this could be done along the lines of Appendix \ref{appendix:proofAFL}. This
would require generalizing Hogervorst's dimensional reduction formula
{\cite{Hogervorst:2016hal}} to unequal external operators.

Using {\eqref{gen-appr}}, the argument could be completed as follows. Let
$\delta > 0$ be an arbitrary positive number. Consider conformal blocks of
operators of twist $\tau$ between $2 \nu + \delta$ and $\Delta_1+\Delta_2 -
\varepsilon$. The factor $\frac{1}{\tau - 2 \nu}$ in the second term in the
denominator of {\eqref{gen-appr}} is then uniformly bounded, as is the ratio
$F_{\tau + 2, \ell}^{2112} (z) / F_{\tau, \ell}^{2112} (z)$. The second term
in the denominator is thus always subleading in this range of $\tau$ and can
be dropped. We can argue exactly as in Section \ref{crossed-low} that their
total crossed-channel contribution of these twists vanishes in the
DLC$_{\ast}$ limit.

Next consider conformal blocks of operators of $\tau \in (2 \nu, 2 \nu +
\delta)$. According to {\eqref{gen-appr}}, each of these conformal blocks is
essentially equivalent, within a constant, to a sum of a collinear block of
twist $\tau$ and a collinear block of twist $\tau + 2$. Let us treat them separately (so that their
relative coefficient does not matter). We will call them ``pieces of type A and
B''. Type A pieces have twist below $\Delta_1+\Delta_2 -
	\varepsilon$ (for
sufficiently small $\delta, \varepsilon$, and $\phi_1$, $\phi_2$ strictly above the unitarity bound) so we can use the argument of
Section \ref{crossed-low} to show that their total crossed-channel
contribution vanishes in the DLC$_{\ast}$ limit. Type B pieces have twist in
the interval $(d, d + \delta)$. Suppose that $\Delta_1 + \Delta_2 \neq d .$
Then for sufficiently small $\varepsilon$ and $\delta$, type B pieces will all
have twist outside the interval $\Delta_1 + \Delta_2 \pm \varepsilon$. In
other words, they will be either high twist or low twist. Depending on which
case is realized, we use arguments of Sections \ref{crossed-high} or
\ref{crossed-low} to show that their total crossed-channel contribution also
vanishes in the DLC$_{\ast}$ limit.

We can then finish the argument as in Section \ref{crossed-around}.

It appears that the restriction $\Delta_1 + \Delta_2 \neq d$ is significant.
For $\Delta_1 + \Delta_2 = d$, level-1 descendants of primaries approaching
the unitarity bound could also reproduce the needed behavior in the
DLC$_{\ast}$ limit. We do not have an example of a theory where this happens,
but our analysis leaves room for this possibility. This would be a
manifestation of the ``fake primary effect'' discussed in
{\cite{Karateev:2019pvw}}.

\textit{Spinning operators.} Generalizations of Theorem
\ref{theorem:DLCtwist} for spinning external operators were also considered,
at an intuitive level, in {\cite{Fitzpatrick:2015qma}}, see also
{\cite{Li:2015itl}}. Their claim is that for arbitrary spin operators
$\mathcal{O}_1, \mathcal{O}_2$ with twists $\tau_1, \tau_2$, there is a
sequence of operators with twists asymptoting to $\tau_1 + \tau_2$ at large
spin. It may be possible to give a rigorous proof of this general result.
Generalizing Lemma \ref{AFLz} to spinning conformal blocks is likely to be the
most subtle step.

\section{Lightcone modular bootstrap}\label{LMB}

As mentioned in the introduction, unitary 2D CFTs exhibit, under
some conditions, twist accumulation phenomena for Virasoro primaries. A
rigorous treatment of these phenomena will be given here. The logic will be analogous to Section \ref{LC}, with crossing symmetry replaced by the modular invariance of the CFT
partition function. In this section we prefer not to draw too much on this analogy, not to clutter the
presentation. The analogy will be discussed in detail in Section \ref{GF} below.

We consider a unitary, modular invariant 2D CFT with central
charge ${c > 1}$ and a positive twist gap
$\tau_{\text{gap}} > 0$ in the spectrum of primaries.\footnote{In this
	section ``primary'' stands for Virasoro primary.} We consider the torus
partition function $Z (\beta, \bar{\beta})$ defined by
\begin{equation}
	Z (\beta, \bar{\beta}) \equiv \text{Tr}_{H_{\text{CFT}}} \left( e^{- \beta
		\left( L_0 - \frac{c}{24} \right)} e^{- \bar{\beta} \left( \bar{L}_0 -
		\frac{c}{24} \right)} \right) . \label{Z}
\end{equation}
Thus $\beta$ and $\bar{\beta}$ are the inverse temperatures for the left
movers and the right movers respectively.

The trace in {\eqref{Z}} is over the CFT Hilbert space whose basis is formed
by primary representations $V_h \otimes V_{\bar{h}}$ with $h, \bar{h}$
belonging to the spectrum of the theory. We assume that the unit operator $V_0
\otimes V_0$ is in the spectrum. The twist gap assumption means that
\begin{equation}
	h, \bar{h} \geqslant \tau_{\rm gap} / 2 \label{twistgap2d}
\end{equation}
for all representations rather than the unit operator.

The arguments below work equally well for the discrete and the continuum
spectrum case. The sum over primaries will be denoted as $\sum_{h, \bar{h}}$.
For a continuum spectrum this has to be replaced by $\int \nospace d \mu
(h, \bar{h})$ where $\mu$ is a non-negative spectral measure.

For $c > 1$, the characters of Virasoro (unitary) representations are given
by
\begin{equation}
	\chi_h (\beta) \equiv \text{Tr}_{V_h} \left( e^{- \beta \left( L_0 -
		\frac{c}{24} \right)}  \right) = \frac{e^{\frac{c - 1}{24} \beta}}{\eta
		(\beta)} \times \begin{cases}
		1 - e^{- \beta} &\text{if } h = 0,\\
		e^{- \beta h}&\text{if } h > 0.
	\end{cases}
\end{equation}

The Dedekind function $\eta$ accounts for the contribution of descendants.
Then the partition function $Z (\beta, \bar{\beta})$ can be written as
\begin{equation}
	Z (\beta, \bar{\beta}) = \frac{e^{(\beta + \bar{\beta}) \frac{c -
				1}{24}}}{\eta (\beta) \eta (\bar{\beta})} \tilde{Z} (\beta,
	\bar{\beta}),
\end{equation}
where
\begin{equation}
	\tilde{Z} (\beta, \bar{\beta}) = (1 - e^{- \beta}) (1 - e^{-
		\bar{\beta}}) + \sum_{h, \bar{h}} e^{- \beta h - \bar{\beta} \bar{h}}
	\label{def:Zvir} .
\end{equation}
We call $\tilde{Z}$ the reduced partition function. We will refer to the terms in this expansion as ``blocks.''

We assume that the partition function is finite when $\beta, \bar{\beta} \in
(0, + \infty)$. We also assume that the partition function is modular
invariant. We will only use invariance under the modular S transformation,
which reads:
\begin{equation}
	Z (\beta, \bar{\beta}) = Z_{\text{}} \left( \frac{4 \pi^2}{\beta}, \frac{4
		\pi^2}{\bar{\beta}} \right) . \label{Smod}
\end{equation}
Using the behavior of the $\eta$ function under the S transformation
\begin{equation}
	\eta (\beta) = \sqrt{\frac{2 \pi}{\beta}} \eta \left( \frac{4 \pi^2}{\beta}
	\right),
\end{equation}
we find that $\tilde{Z}$ transforms as
\begin{equation}
	\tilde{Z} (\beta, \bar{\beta}) = K (\beta, \bar{\beta}) \tilde{Z}
	\left( \frac{4 \pi^2}{\beta}, \frac{4 \pi^2}{\bar{\beta}} \right)
	\label{modulartransformation},
\end{equation}
where
\begin{equation}
	K (\beta, \bar{\beta}) = \sqrt{\frac{4 \pi^2}{\beta \bar{\beta}}} e^{
		4 \pi^2A / \beta} e^{4 \pi^2A / \bar{\beta}} e^{- A\beta} e^{- A\bar{\beta}},
	\label{def:modularK}
\end{equation}
and we have denoted $A \equiv \frac{c - 1}{24}$.

The convergent expansion (\ref{def:Zvir}) and the modular invariance condition
(\ref{modulartransformation}) will be the starting points of our analysis.

In this setup, we prove the following theorem

\begin{theorem}
	\label{boundtwist}Take any unitary, modular invariant 2D CFT,
	having a normalizable vacuum, central charge $c > 1$ and a positive twist
	gap $\tau_{\rm gap}$ in the spectrum of nontrivial primaries, see Eq.
	{\eqref{twistgap2d}}. Then:
	
	(a) $\tau_{\rm gap} \leqslant \frac{c - 1}{12}$;
	
	(b) There is at least one family of primaries $\mathcal{O}_i$ such that $h_i
	\rightarrow \frac{c - 1}{24}$ while $\bar{h}_i \rightarrow \infty$. The same
	is true with $h$ and $\bar{h}$ interchanged.
\end{theorem}

Note the close analogy with Theorem \ref{theorem:DLCtwist}. As there, part (b)
is stronger than part (a), but in the course of the proof it will be important
to establish part (a) first.

This result was previously found in
{\cite{Collier:2016cls,Afkhami-Jeddi:2017idc,Benjamin:2019stq}} based on
intuitive non-rigorous arguments, as we review in Section \ref{int-mod} below.

\begin{remark}
	Condition $c > 1$ is important: unitary minimal models violate this
	assumption and the conclusions of the theorem do not hold. Condition
	$\tau_{\rm gap} > 0$ is also important. To see this consider several copies of
	the unitary minimal models so that the total central charge is above 1. The
	conclusions of the theorem do not hold. It's easy to see that in this case
	there are nontrivial primaries having twist zero.
	
	Unitarity and normalizable vacuum will play an important role in our proof,
	yet we don't know of explicit examples showing that if one drops one of
	these conditions, the theorem does not hold. On the contrary, we know some
	examples which do violate these conditions, but the conclusions still hold.
	This is the Liouville partition function {\cite{Seiberg:1990eb}} which does
	not contain normalizable vacuum, and the Maloney-Witten partition function
	{\cite{Maloney:2007ud,Benjamin:2019stq}} which violates unitarity since it
	has a partly negative spectral density.
\end{remark}

\subsection{Sketch of the proof }

To prove this theorem we will probe the partition function in the limit
$\beta \rightarrow 0$, $\bar{\beta} \rightarrow \infty$.
Since this is a double limit, a natural question is whether the two limits are
taken independently or not. We will find it necessary to take the two limits
not independently. Let us denote $T = \tau_{\rm gap} / 2$. Define the quantity
\begin{equation}
	\mathfrak{b}=\mathfrak{b} (\beta, \bar{\beta}) = \bar{\beta} - \frac{4 \pi^2
		A}{T \beta} + \frac{3}{2 T} \log \beta . \label{bquant}
\end{equation}
Note that $\mathfrak{b} (\bar{\beta}, \beta) < \bar{\beta}$ as $\beta
\rightarrow 0$. We will have to require that not only $\beta \rightarrow 0$,
$\bar{\beta} \rightarrow \infty$, but also that $\bar{\beta} \rightarrow
\infty$ sufficiently fast so that $\mathfrak{b} \rightarrow \infty$. We call
this the ``modular double lightcone limit'', and denote it M$_{\ast}$:
\begin{equation}
	\text{M$_{\ast}$ limit:}\quad \beta \rightarrow 0,
	\quad \bar{\beta} \rightarrow \infty,
	\quad \mathfrak{b} \rightarrow \infty . \label{M*}
\end{equation}
The origin of the quantity $\mathfrak{b} (\bar{\beta}, \beta)$ and of the
$\log \beta$ term in it will become clear in Section
\ref{section:modulardirect}. The M$_{\ast}$ limit is closely analogous to the
DLC$_{\ast}$ limit from Section \ref{LC}, as will be discussed in detail in
Section \ref{GF} below.

The proof consists in several steps:
\begin{enumerate}
	\item \textit{Direct channel} (Section \ref{section:modulardirect}). We
	first study expansion (\ref{def:Zvir}) of the reduced partition function and
	we show that the vacuum term (the first term in (\ref{def:Zvir})) dominates
	in the M$_{\ast}$ limit (\ref{M*}). Namely, we show
	\begin{equation}
		\underset{\text{M}_{\ast}}{\lim} \frac{\tilde{Z} (\beta,
			\bar{\beta})}{\beta} = 1. \label{DLCpart}
	\end{equation}
	Here $\beta$ is the asymptotic behavior of the vacuum term. We will thus we show that
	the total contribution from non-vacuum states
	in this channel is suppressed compared to the vacuum in the M$_{\ast}$
	limit.
	
	\item \textit{Twist upper bound} (Section \ref{gapupper}) Using a slightly
	weaker limit $\beta \rightarrow 0$, $\bar{\beta} \rightarrow \infty$ with
	$\mathfrak{b}= 0$ (see (\ref{bquant})), we arrive at the part (a) of Theorem
	\ref{boundtwist}. We are going to use this bound in next steps.
	
	\item \textit{Crossed channel: individual blocks} (Section
	\ref{section:modularindividual}) Here we show that in the $\text{M}_{\ast}$
	limit, the non-vacuum contribution is indeed subleading with respect to the
	vacuum, which goes as $\beta$.
	
	\item \textit{Crossed channel: high h} (Section
	\ref{section:modularcrossed-high}). Here we turn our attention to the
	crossed channel, and show that, for any $\varepsilon > 0$, the total
	contribution of operators with \ $h \geqslant A + \varepsilon$ vanishes in
	the M$_{\ast}$ limit. This is analogous to high twist sector analysis in
	Section \ref{crossed-high}.
	
	\item \textit{Crossed channel: low h }(Section
	\ref{modularcrossed-low}). We then consider low $h$ and show that the same
	is true for them: for any $\varepsilon > 0$, the total contribution of
	operators with $h \leqslant A - \varepsilon$ also vanishes in the M$_{\ast}$
	limit. This is analogous to low twist sector analysis in Section
	\ref{crossed-low}.
\end{enumerate}
The proof is then completed in Section \ref{section:modularend}.

\subsection{Detailed proof}

\subsubsection{Key lemma}

Much of the arguments below will depend on the following key lemma, analogous
to Prop. \ref{lemmahigh} from the lightcone bootstrap. For any positive
$h_{\ast}$ we define the partial sums $\tilde{Z}_{h \geqslant h_{\ast}}$ and
$\tilde{Z}_{\bar{h} \geqslant h_{\ast}}$ by
\begin{eqnarray}
	\tilde{Z}_{h \geqslant h_{\ast}} (\beta, \bar{\beta}) & = & \sum_{h
		\geqslant h_{\ast}} \sum_{\bar{h}} e^{- h \beta - \bar{h}  \bar{\beta}}, 
	\label{hgtr}\\
	\tilde{Z}_{\bar{h} \geqslant h_{\ast}} (\beta, \bar{\beta}) & = & \sum_h
	\sum_{\bar{h} \geqslant h_{\ast}} e^{- h \beta - \bar{h}  \bar{\beta}} . 
\end{eqnarray}
\begin{lemma}
	\label{keylemma}For any $\beta_0 \in (0, \infty)$, there exists a finite
	$C $ (which depends on the CFT and on $\beta_0, h_*$) such that the following bounds hold in their shown
	respective ranges of $\beta, \bar{\beta}$:
	\begin{eqnarray}
		\tilde{Z}_{h \geqslant h_{\ast}} (\beta, \bar{\beta}) & \leqslant & C
		 \bar{\beta}^{- 1 / 2} e^{4 \pi^2 A / \bar{\beta} - h_{\ast}
			\beta} \qquad (\beta \geqslant \beta_0 \geqslant \bar{\beta}), 
		\label{keylemma1}\\
		\tilde{Z}_{\bar{h} \geqslant h_{\ast}} (\beta, \bar{\beta}) & \leqslant &
		C  \beta^{- 1 / 2} e^{4 \pi^2 A / \beta - h_{\ast}
			\bar{\beta}} \qquad (\bar{\beta} \geqslant \beta_0 \geqslant \beta) . 
		\label{keylemma2}
	\end{eqnarray}
\end{lemma}

\begin{proof}
	We will show the second bound. The first bound is fully analogous with the
	interchange of $h, \beta$ with $\bar{h}, \bar{\beta}$. We have (see
	explanations below)
	\begin{eqnarray}
			\tilde{Z}_{\bar{h} \geqslant h_{\ast}} (\beta, \bar{\beta}) & \leqslant
			& e^{- h_{\ast} (\bar{\beta} - \beta_0)} \sum_h \sum_{\bar{h} \geqslant
				h_{\ast}} e^{- h \beta - \bar{h} \beta_0}\nonumber\\
			& \leqslant & e^{- h_{\ast} (\bar{\beta} - \beta_0)} \tilde{Z} (\beta,
			\beta_0)\nonumber\\
			& = & e^{- h_{\ast} (\bar{\beta} - \beta_0)} K (\beta, \beta_0) 
			\tilde{Z} \left( \frac{4 \pi^2}{\beta}, \frac{4 \pi^2}{\beta_0} \right)
			. \label{modular:directnonvac}
		\end{eqnarray} 
	Here in the first line we used that $\bar{h} \geqslant h_{\ast}$ and
	$\bar{\beta} \geqslant \beta_0$. In the second line we bounded
	$\sum_h \sum_{\bar{h} \geqslant h_{\ast}} e^{- h \beta - \bar{h} \beta_0}$ by
	the full $\tilde{Z} (\beta, \beta_0)$. In the third line we used the modular
	invariance.
	
	We now use Eq.\,(\ref{def:Zvir}) for $\tilde{Z} \left( \frac{4
		\pi^2}{\beta}, \frac{4 \pi^2}{\beta_0} \right)$:
	\begin{equation}
		\tilde{Z} \left( \frac{4 \pi^2}{\beta}, \frac{4 \pi^2}{\beta_0} \right) =
		(1 - e^{- 4 \pi^2 / \beta}) (1 - e^{- 4 \pi^2 / \beta_0}) + \sum_{h,
			\bar{h}} e^{- 4 \pi^2 h / \beta - 4 \pi^2 \bar{h} / \beta_0} .
	\end{equation}
	We estimate the first term by 1 and all the other terms by their value at
	$\beta = \beta_0$, using monotonicity and $\beta \leqslant \beta_0$:

	\begin{equation}
		\tilde{Z} \left( \frac{4 \pi^2}{\beta}, \frac{4 \pi^2}{\beta_0} \right)
		\leqslant 1 + \sum_{h, \bar{h}} e^{- 4 \pi^2 h / \beta_0 - 4 \pi^2 \bar{h}
			/ \beta_0} \backassign C_1. \label{sqbracket}
	\end{equation}
	Since we are assuming that the partition function is finite for any $0 <
	\beta, \bar{\beta} < \infty$, the sum in r.h.s.~of {\eqref{sqbracket}} is
	finite, and $C_1 < \infty$ (this constant thus depends on the CFT and on $\beta_0$).
	
	Plugging {\eqref{sqbracket}} back into {\eqref{modular:directnonvac}}, we
	get
	\begin{eqnarray}
		\tilde{Z}_{\bar{h} \geqslant h_{\ast}} (\beta, \bar{\beta}) & \leqslant &
		C_1  e^{- h_{\ast} (\bar{\beta} - \beta_0)} K (\beta, \beta_0)\nonumber\\
		& \equiv & C_1 e^{- h_{\ast} (\bar{\beta} - \beta_0)}
		\sqrt{\frac{4 \pi^2}{\beta \beta_0}} e^{A (4 \pi^2 / \beta + 4 \pi^2 /
			\beta_0 - \beta - \beta_0)}\nonumber\\
		& \leqslant & C \beta^{- 1 / 2} e^{4 \pi^2 A / \beta -
			h_{\ast} \bar{\beta}},
	\end{eqnarray}
	where $C = C_1  e^{- (A - h_{\ast}) \beta_0} e^{4
		\pi^2 A / \beta_0} \sqrt{\frac{4 \pi^2}{\beta_0}}$ and we used $e^{- A
		\beta} \leqslant 1$ when passing to the third line. The lemma is proved.
\end{proof}

\subsubsection{Direct channel}\label{section:modulardirect}

Here we will prove Eq.\,{\eqref{DLCpart}}. The vacuum contribution in
$\tilde{Z}$ goes as $\beta$ in the M$_{\ast}$ limit. We need to show that the
total non-vacuum contribution is subleading compared to the vacuum. Recall that
we denoted $T = \tau_{\rm gap} / 2$, and that we have $h \geqslant T, \bar{h}
\geqslant T$ for all non-vacuum terms.

We choose and fix a $\beta_0 \in (0, \infty)$. Its precise value does not
matter; e.g. $\beta_0 = 1$ would work. Since in the M$_{\ast}$ limit we have
$\beta \rightarrow 0$, $\bar{\beta} \rightarrow \infty$, eventually we will
have that
\begin{equation}
	\bar{\beta} \geqslant \beta_0 \geqslant \beta . \label{regime}
\end{equation}
Note that the non-vacuum contribution can be written as
\begin{equation}
	[\tilde{Z} (\beta, \bar{\beta})]_{\text{nonvac}} = \tilde{Z}_{\bar{h}
		\geqslant T} (\beta, \bar{\beta}) .
\end{equation}
We can therefore use the estimate {\eqref{keylemma2}} from Lemma
\ref{keylemma} with $h_{\ast} = T$. We have:
\begin{eqnarray}
	{}[\tilde{Z} (\beta, \bar{\beta})]_{\text{nonvac}} & \leqslant & C
	 \beta^{- 1 / 2} e^{4 \pi^2 A / \beta - T \bar{\beta}} \equiv C \beta\, e^{- T\mathfrak{b}} \label{esti:modulardirect}, 
\end{eqnarray}
where we used the definition \eqref{bquant} of $\mathfrak{b}$. Since $\mathfrak{b}
\rightarrow \infty$ in the $\text{M}_{\ast}$ limit, we conclude that the
non-vacuum contribution is indeed subleading with respect to the vacuum, which
goes as $\beta$.

\subsubsection{Twist gap upper bound}\label{gapupper}

In this section we will establish an upper bound on the twist gap, namely that
$\tau_{\rm gap} \leqslant \frac{c - 1}{12}$, or $T \leqslant A$ in our notation.
This is Theorem \ref{boundtwist}(a), the modular-bootstrap analogue of Theorem
\ref{theorem:DLCtwist}(a). An argument for this twist gap upper bound, due to
Tom Hartman, was presented in the modular bootstrap literature
{\cite{Collier:2016cls,Afkhami-Jeddi:2017idc,Benjamin:2019stq}}. Those
considerations are not rigorous, although the argument of Ref.\,{\cite{Benjamin:2019stq}} can be made rigorous, see Section \ref{int-mod}
below. Our proof here is different from these papers, but similar to the
proof of Theorem \ref{theorem:DLCtwist}(a), see the end of Section
\ref{direct}.

\begin{proof}[Proof of Theorem \ref{boundtwist}(a)] To prove that $T\leqslant A$, it turns out useful to work not in the M$_{\ast}$ limit but in the limit
	\begin{equation}
		\beta \rightarrow 0, \quad \nospace \bar{\beta} = \frac{4 \pi^2 A}{T
			\beta} - \frac{3}{2 T} \log \beta \rightarrow \infty, \label{newlimit}
	\end{equation}
	We see that this limit corresponds to $\mathfrak{b}= 0$, while the M$_{\ast}$ limit has $\mathfrak{b}\rightarrow\infty$.\footnote{The limit $\beta \rightarrow 0$, $\beta \rightarrow \infty$, $\mathfrak{b}= C$ with
		a nonzero $C$ would also work for the purpose of showing $T\leqslant A$.} 
	
	Consider how both sides of
	(\ref{modulartransformation}) behave in the limit \eqref{newlimit}. We have shown in the previous section that $\tilde{Z} (\beta, \bar{\beta})
	= \beta + o (\beta)$ in the M$_{\ast}$ limit. Let us examine what changes in
	the new limit {\eqref{newlimit}}. Clearly, the vacuum contribution still
	goes as $\beta + o (\beta)$. For the non-vacuum contribution, we use the
	estimate {\eqref{esti:modulardirect}}. In the previous section we used that
	estimate for $\mathfrak{b} \rightarrow \infty$, but here we use it for
	$\mathfrak{b}= 0$. That estimate implies that for $\mathfrak{b}= 0$ the
	non-vacuum contribution is $O (\beta)$. Thus we conclude that:
	\begin{equation}
		\tilde{Z} (\beta, \bar{\beta}) = O (\beta) \text{\qquad in the limit
			{\eqref{newlimit}}.} \label{boundnew}
	\end{equation}
	Consider now the crossed channel, i.e.\,r.h.s. of Eq.\,(\ref{modulartransformation}). Keeping only the vacuum contribution, we have
	the lower bound:
	\begin{eqnarray}
		\tilde{Z} (\beta, \bar{\beta}) & \geqslant & K (\beta, \bar{\beta}) (1 -
		e^{- 4 \pi^2 / \beta}) (1 - e^{- 4 \pi^2 / \bar{\beta}}) \nonumber\\
		& \equiv & \sqrt{\frac{4 \pi^2}{\beta \bar{\beta}}} e^{4 \pi^2 A /
			\beta} e^{4 \pi^2 A / \bar{\beta}} e^{- A \beta} e^{- A \bar{\beta}} (1 -
		e^{- 4 \pi^2 / \beta}) (1 - e^{- 4 \pi^2 / \bar{\beta}}) . 
		\label{lowercrossed}
	\end{eqnarray}
	Let us examine how this expression goes in the limit {\eqref{newlimit}}.
	Using $e^{4 \pi^2 A / \bar{\beta}} \rightarrow 1$, $e^{- A \beta}
	\rightarrow 1$, $1 - e^{- 4 \pi^2 / \beta} \rightarrow 1$, $1 - e^{- 4 \pi^2
		/ \bar{\beta}} = 4 \pi^2 / \bar{\beta} + O (\bar{\beta}^{- 2})$, the product
	of all those factors can be bounded from below by (any number below 1 would
	work in place of 1/2)
	\begin{equation}
		(1/2) \times 4 \pi^2 / \bar{\beta} .
	\end{equation}
	Furthermore, we eliminate the factor $e^{4 \pi^2 A / \beta}$ in
	{\eqref{lowercrossed}} by the relation
	\begin{equation}
		\beta^{- 3 / 2} e^{4 \pi^2 A / \beta - T \bar{\beta}} = e^{-
			T\mathfrak{b}}, \label{brel}
	\end{equation}
	and using $\mathfrak{b}= 0$. Collecting the factors, we obtain the lower
	bound:
	\begin{equation}
		\tilde{Z} (\beta, \bar{\beta}) \geqslant \beta \sqrt{\frac{4 \pi^2}{
				\bar{\beta}}} \frac{2 \pi^2}{\bar{\beta}} e^{(T - A) \bar{\beta}} .
	\end{equation}
	From this bound we see that $\tilde{Z} (\beta, \bar{\beta}) / \beta$ blows
	up in the considered limit for $T > A$. Since we know from
	{\eqref{boundnew}} that $\tilde{Z} (\beta, \bar{\beta}) / \beta$ should be
	bounded, we conclude that $T \leqslant A$.
\end{proof}

\subsubsection{Crossed channel: individual
	blocks}\label{section:modularindividual}

From this section on, we will study the crossed channel contributions in the
$M_{\ast}$ limit. We would like to see how their asymptotics compare to what
we already know from Section \ref{section:modulardirect}, namely that
$\tilde{Z} (\beta, \bar{\beta}) \sim \beta$ in this limit.

By modular invariance, the sum of all crossed channel contributions should combine to give
$\tilde{Z} (\beta, \bar{\beta}) \sim \beta$. In this section we will show that any
individual block contribute $o (\beta)$. This will imply in particular that there must be infinitely many
blocks present in the reduced partition function.

We use the simple fact that each block is bounded by 1:
\begin{eqnarray}
	\text{vacuum} : &  & (1 - e^{- \beta}) (1 - e^{- \bar{\beta}}) \leqslant 1,
	\nonumber\\
	\text{non-vacuum} : &  & e^{- h \beta - \bar{h} \bar{\beta}} \leqslant 1. 
\end{eqnarray}
So to see that each single block contributes $o (\beta)$ to the r.h.s.~of
(\ref{modulartransformation}), it suffices to check that
\begin{equation}
	K (\beta, \bar{\beta}) = o (\beta)\qquad \text{in the M$_{\ast}$ limit.}
	\label{needK}
\end{equation}
Using $e^{4 \pi^2 A / \bar{\beta}} \rightarrow 1$, $e^{- A \beta} \rightarrow
1$, we estimate the product of these two factors in $K(\beta,\bar\beta)$, Eq.~\eqref{def:modularK}, from above by 2. (Any number
above 1 would work.) We use the relation {\eqref{brel}} to eliminate the
factor $e^{4 \pi^2 A/ \beta}$ in $K(\beta,\bar\beta)$. We obtain:
\begin{equation}
	K (\beta, \bar{\beta}) \leqslant \beta \sqrt{\frac{4 \pi^2}{
			\bar{\beta}}} 2\, e^{(T - A) \bar{\beta}} \leqslant 2 \beta \sqrt{\frac{4
			\pi^2}{ \bar{\beta}}} .
\end{equation}
where we used $T \leqslant A$ (part (a) of the theorem, proved in Section
\ref{gapupper}). Eq.~{\eqref{needK}} follows.

\subsubsection{Crossed channel: high $h$}\label{section:modularcrossed-high}

In this section we will show that, for any $\varepsilon > 0$, twists $h
\geqslant A + \varepsilon$ give a subleading contribution in the crossed
channel. This means that:
\begin{equation}
	K (\beta, \bar{\beta}) \tilde{Z}_{h \geqslant h_{\ast}} \left( \frac{4
		\pi^2}{\beta}, \frac{4 \pi^2}{\bar{\beta}} \right) \label{modular:hightwist}
	= o (\beta) \quad \text{in the M$_{\ast}$ limit,}
\end{equation}
for $h_{\ast} = A + \varepsilon$, where $\tilde{Z}_{h \geqslant h_{\ast}}$ is
defined as in {\eqref{hgtr}}.

As in Section \ref{section:modulardirect} we will work in the regime
$\bar{\beta} \geqslant \beta_0 \geqslant \beta$. Note that
\begin{equation}
	\frac{4 \pi^2}{\beta} \geqslant \frac{4 \pi^2}{\beta_0} \geqslant \frac{4
		\pi^2}{\bar{\beta}} .
\end{equation}
We are therefore in a position to bound $\tilde{Z}_{h \geqslant h_{\ast}}
\left( \frac{4 \pi^2}{\beta}, \frac{4 \pi^2}{\bar{\beta}} \right)$ using Lemma
\ref{keylemma}, Eq.\,{\eqref{keylemma1}}. We obtain:
\begin{equation}
	\tilde{Z}_{h \geqslant h_{\ast}} \left( \frac{4 \pi^2}{\beta}, \frac{4
		\pi^2}{\bar{\beta}} \right) \leqslant C' \bar{\beta}^{1 / 2} e^{A
		\bar{\beta} - 4 \pi^2 h_{\ast} / \beta},
\end{equation}
where $C' = C (4 \pi^2 / \beta_0) / 2 \pi$.

Using this bound in the l.h.s.~of {\eqref{modular:hightwist}}, and using the
definition of $K (\beta, \bar{\beta})$, we obtain
\begin{eqnarray}
	\text{l.h.s.~of {\eqref{modular:hightwist}}} & \leqslant & C' \sqrt{\frac{4
			\pi^2}{\beta}} e^{(A - h_{\ast}) 4 \pi^2 / \beta} e^{A (4 \pi^2 /
		\bar{\beta} - \beta)} . 
\end{eqnarray}
The important factor here is $e^{(A - h_{\ast}) 4 \pi^2 / \beta} \equiv e^{-
	\varepsilon \, 4 \pi^2 / \beta}$. Because of this factor the r.h.s.~of
{\eqref{modular:hightwist}} vanishes exponentially fast as $\beta \rightarrow
0$. In particular Eq.~{\eqref{modular:hightwist}} holds. Note that we did not
use the full strength of the M$_{\ast}$ limit in this section. We only needed
$\beta \rightarrow 0$, $\bar{\beta} \geqslant \beta_0$.

\subsubsection{Crossed channel: low $h$}\label{modularcrossed-low}

In this section we will show that the total contribution of twists $h
\leqslant t_{\ast} := A - \varepsilon$ is subleading in the crossed channel for
any $\varepsilon > 0$. We consider the partial sum

\begin{equation}
	\tilde{Z}_{h \leqslant t_{\ast}} (\beta, \bar{\beta}) \assign \sum_{h
		\leqslant t_{\ast}} \sum_{\bar{h}} e^{- \beta h - \bar{\beta} \bar{h}}
	. \label{lowblocks2d}
\end{equation}
For convenience this does not include the vacuum contribution. We know anyway
from Section \ref{section:modularindividual} that any individual contribution,
in particular the vacuum, is subleading. The point here is to show this for
the sum {\eqref{lowblocks2d}}, which includes potentially infinitely many
blocks.

The first idea will be to estimate $\tilde{Z}_{h \leqslant t_{\ast}} (\beta,
\bar{\beta})$ in terms of $\tilde{Z} (\beta', \bar{\beta})$. The precise
value of $\beta'$ will be fixed below. We have, for $\beta' \geqslant \beta$,
\begin{equation}
	e^{- \beta h} \leqslant e^{- \beta' h + h (\beta' - \beta)} \leqslant
	e^{t_{\ast} (\beta' - \beta)} e^{- \beta' h}\quad
	(h \leqslant t_{\ast}) .
\end{equation}
Multiplying this by $e^{- \bar{\beta} \bar{h}}$ and summing over all nonvacuum
contributions with $h \leqslant t_{\ast}$, we get
\begin{equation}
	\tilde{Z}_{h \leqslant t_{\ast}} (\beta, \bar{\beta}) \leqslant e^{t_{\ast}
		(\beta' - \beta)} \tilde{Z}_{h \leqslant t_{\ast}} (\beta', \bar{\beta})
	\leqslant e^{t_{\ast} (\beta' - \beta)} \tilde{Z} (\beta', \bar{\beta})
	\qquad (0 < \beta \leqslant \beta'),
\end{equation}
where in the last step we bounded $\tilde{Z}_{h \leqslant t_{\ast}}$ by
$\tilde{Z}$. Moving to the crossed channel variables $\beta \rightarrow 4
\pi^2 / \beta$, $\bar{\beta} \rightarrow 4 \pi^2 / \bar{\beta}$, we get
\begin{equation}
	\tilde{Z}_{h \leqslant t_{\ast}} \left( \frac{4 \pi^2}{\beta}, \frac{4
		\pi^2}{\bar{\beta}} \right) \leqslant e^{t_{\ast} (4 \pi^2 / \beta' - 4
		\pi^2 / \beta)} \tilde{Z} \left( \frac{4 \pi^2}{\beta'}, \frac{4
		\pi^2}{\bar{\beta}} \right)\qquad (0 < \beta'
	\leqslant \beta) . \label{modular:lowtwistcross}
\end{equation}
We plug this into the r.h.s.~of the modular invariance equation
(\ref{modulartransformation}), and get for all $0 < \beta' \leqslant \beta$,
\begin{eqnarray}
	K (\beta, \bar{\beta}) \tilde{Z}_{h \leqslant t_{\ast}} \left( \frac{4
		\pi^2}{\beta}, \frac{4 \pi^2}{\bar{\beta}} \right) & \leqslant & K (\beta,
	\bar{\beta}) e^{t_{\ast} (4 \pi^2 / \beta' - 4 \pi^2 / \beta)} \tilde{Z}
	\left( \frac{4 \pi^2}{\beta'}, \frac{4 \pi^2}{\bar{\beta}} \right)
	\nonumber\\
	& = & \frac{K (\beta, \bar{\beta})}{K (\beta', \bar{\beta})} e^{t_{\ast} (4
		\pi^2 / \beta' - 4 \pi^2 / \beta)} \tilde{Z} (\beta', \bar{\beta})
	\nonumber\\
	& = & \sqrt{\frac{\beta'}{\beta}} e^{- \varepsilon (4 \pi^2 / \beta' - 4
		\pi^2 / \beta)} e^{- A (\beta - \beta')} \tilde{Z} (\beta', \bar{\beta})
	\nonumber\\
	& \leqslant & e^{- \varepsilon (4 \pi^2 / \beta' - 4 \pi^2 / \beta)}
	\tilde{Z} (\beta', \bar{\beta}),  \label{lt}
\end{eqnarray}
where we used (\ref{modular:lowtwistcross}) in the first line, modular
invariance in the second line, the explicit expression of $K$ (see
Eq.\,(\ref{def:modularK})) and $\varepsilon = A - t_{\ast}$ in the third line,
and $\beta' \leqslant \beta$ in the last line.

Now we come to the second idea. We will pick $\beta'$ somewhat smaller than
$\beta$. Then, the first factor in the r.h.s.~of {\eqref{lt}} is exponentially
suppressed. On the other hand $\tilde{Z} (\beta', \bar{\beta}) > \tilde{Z}
(\beta, \bar{\beta})$. We would like that the first effect overcomes the
second. For this we choose $\beta'$ from the condition
\begin{equation}
	\mathfrak{b} (\beta', \bar{\beta}) \equiv \bar{\beta} - f (\beta') = 0,
	\qquad f (\beta') \assign \frac{4 \pi^2 A}{T \beta'} - \frac{3}{2 T} \log
	\beta' . \label{choice}
\end{equation}
This condition guarantees that
\begin{equation}
	\tilde{Z} (\beta', \bar{\beta}) = O (\beta') . \label{Ztildebound}
\end{equation}
This is the bound {\eqref{boundnew}} from the proof of Theorem
\ref{boundtwist}(a), with $\beta'$ instead of $\beta$.

Compare next {\eqref{choice}} to the M$_{\ast}$ limit condition, which can be
written as
\begin{equation}
	\mathfrak{b} (\beta, \bar{\beta}) \equiv \bar{\beta} - f (\beta) \rightarrow
	\infty,
\end{equation}
Taking the difference of the two conditions, we have:
\begin{equation}
	\mathfrak{b} (\beta, \bar{\beta}) = f (\beta') - f (\beta) \rightarrow
	\infty \label{bdiff} .
\end{equation}
The function $f (\beta)$ is monotonically decreasing. This implies in
particular that eventually $\beta' \leqslant \beta$ as needed for
{\eqref{lt}}. Furthermore, we claim that
\begin{equation}
	\frac{1}{\beta'} - \frac{1}{\beta} \rightarrow \infty \label{betabeta} .
\end{equation}
Indeed assume the contrary, i.e.~that$\frac{1}{\beta'} \leqslant
\frac{1}{\beta} + C$ for some constant $C$. We have
\begin{eqnarray}
	f (\beta') - f (\beta) & = & \frac{4 \pi^2 A}{T} \left( \frac{1}{\beta'} -
	\frac{1}{\beta} \right) + \frac{3}{2 T} \left( \log \frac{1}{\beta'} - \log
	\frac{1}{\beta} \right) .  \label{minor-compl}
\end{eqnarray}
The first term is then bounded by $\frac{4 \pi^2 A}{T} C$, while the second
term by $\frac{3}{2 T} C \beta$,\footnote{We use here $\log (x + C) - \log
	(x) = \int_x^{x + C} \frac{d y}{y} \leqslant C / x$.} which goes to zero since
$\beta \rightarrow 0$. We reach a contradiction with {\eqref{bdiff}}, and thus
we conclude that {\eqref{betabeta}} must be true.

Now let us look back at the estimate (\ref{lt}). By {\eqref{Ztildebound}} and
$\beta' \leqslant \beta$ we have $\tilde{Z} (\beta', \bar{\beta}) = O (\beta)$
in that bound. The prefactor $e^{- \varepsilon (4 \pi^2 / \beta' - 4 \pi^2 /
	\beta)} \rightarrow 0$ by {\eqref{betabeta}}. Therefore we proved
\begin{equation}
	K (\beta, \bar{\beta}) \tilde{Z}_{h \leqslant t_{\ast}} \left( \frac{4
		\pi^2}{\beta}, \frac{4 \pi^2}{\bar{\beta}} \right) = o (\beta),
\end{equation}
i.e.\,that the contribution of low twists is subleading in the M$_{\ast}$
limit.

\subsubsection{End of proof of Theorem
	\ref{boundtwist}}\label{section:modularend}

So far we have shown that, in the M$_{\ast}$ limit,
\begin{itemize}
	\item the l.h.s.~of the modular invariance equation (direct channel) behaves as $\sim\beta$;
	
	\item in the r.h.s.~of the modular invariance equation, the total
	contribution of $h$ with $\vert h - A \vert \geqslant \varepsilon$ is $o (\beta)$,
	for any $\varepsilon > 0$;
\end{itemize}
We conclude that, for any $\varepsilon > 0$, \ the total crossed-channel
contribution of $h$ with $\vert h - A \vert < \varepsilon$ should go asymptotically as
$\beta$, to be consistent with the direct channel.

It remains to argue that there is an infinite sequence of
Virasoro primaries $\mathcal{O}_i$ with $h_i \rightarrow A$ and $\bar{h}_i
\rightarrow \infty$. This is quite obvious for the discrete spectrum, and a
bit more subtle for the continuum spectrum. We will give a complete, if a bit
pedantic, argument.

Pick any positive sequence monotonically going to zero $\varepsilon_i \searrow
0$. Consider the set $U_i$ of all primaries with $\vert h - A \vert < \varepsilon_i$.
As $i$ increases, the set $U_i$ is getting smaller. Yet, for any $i$, the
total crossed-channel contribution of this set goes as $\beta$ in the
M$_{\ast}$ limit. In particular, for any $i$, the set $U_i$ is nonempty.
Moreover it is infinite (since the crossed-channel contribution of any finite
set of primaries is subdominant, Section \ref{section:modularindividual}).

We will show below that $\bar{h}$ is unbounded in the set of primaries $U_i$. Assuming
that, the proof is finished as follows. Since $\bar{h}$ is unbounded, we can
pick a primary $\mathcal{O}_i =\mathcal{O}_{h_i, \bar{h}_i} \in U_i$ such that
$\bar{h}_i > i$, while $h_i$ satisfies $\vert h_i - A \vert < \varepsilon_i$, as for
all primaries in $U_i$. Repeating this construction for $i = 1, 2, 3, \ldots$,
we construct the needed sequence.

It remains to show that $\bar{h}$ is unbounded in $U_i$. If the spectrum is
discrete, this follows from the fact that $U_i$ is infinite. Indeed any
infinite set of primaries of bounded $h$ has then unbounded $\bar{h}$.

If the spectrum is continuous, unboundedness of $\bar{h}$ follows from the
following lemma. For any $\Lambda > 0$ define:
\begin{equation}
	\tilde{Z}_{\Lambda} (\beta, \bar{\beta}) \assign \int_{T \leqslant h,
		\bar{h} \leqslant \Lambda} d \mu (h, \bar{h}) e^{- \beta h - \bar{\beta}
		\bar{h}} .
\end{equation}
This is the contribution of all nontrivial primaries up to cutoff $\Lambda$ on
both $h$ and $\bar{h}$. We write this in the general notation allowing for both discrete and continuum
spectrum. The spectral measure $d \mu$ is assumed non-negative.

The following result shows that $\tilde{Z}_{\Lambda}$ gives a subleading
contribution in the crossed channel.

\begin{lemma}
	\label{contMLB}For any $\Lambda > 0$, the crossed-channel contribution of
	all $h, \bar{h} \leqslant \Lambda$ gives a subleading contribution in the
	M$_{\ast}$ limit, i.e.
	\begin{equation}
		K (\beta, \bar{\beta})  \tilde{Z}_{\Lambda} \left( \frac{4 \pi^2}{\beta},
		\frac{4 \pi^2}{\bar{\beta}} \right) = o (\beta) \quad \text{in the M$_{\ast}$ limit.}
	\end{equation}
	This is true for both discrete and continuous spectrum.
\end{lemma}

\begin{proof}
	The argument is a modification of that of Section
	\ref{section:modularindividual}. There, we already proved that that $K
	(\beta, \bar{\beta}) = o (\beta)$ (Eq.\,(\ref{needK})). It remains to show
	that $\tilde{Z}_{\Lambda} (\beta, \bar{\beta})$ is bounded by a constant uniformly
	in $\beta$ and $\bar{\beta}$. Pick an arbitrary $\beta_0 \in (0, \infty)$.
	We have, for arbitrary $\beta, \bar{\beta}\geqslant 0$, a generous estimate:
	\begin{equation}
		e^{- \beta h - \bar{\beta} \bar{h}} \leqslant e^{2 \Lambda \beta_0}
		e^{- \beta_0 (h + \bar{h})} \qquad (h \leqslant \Lambda, \bar{h}
		\leqslant \Lambda) .
	\end{equation}
	Plugging this into the definition of $\tilde{Z}_{\Lambda}$, we get the needed uniform bound
	\begin{equation}
		\tilde{Z}_{\Lambda} (\beta, \bar{\beta}) \leqslant e^{2 \Lambda
			\beta_0} \tilde{Z}_{\Lambda} (\beta_0, \beta_0) \leqslant e^{2 \Lambda
			\beta_0} \tilde{Z} (\beta_0, \beta_0) < \infty .
	\end{equation}
	where in the last step, we bound $\tilde{Z}_{\Lambda} (\beta_0, \beta_0)$ by
	the full partition function $\tilde{Z} (\beta_0, \beta_0)$. 
\end{proof}

Theorem \ref{boundtwist}(b) is now fully proved.

\subsection{Intuitive arguments}\label{int-mod}

In this section we will mention the intuitive arguments which were previously
used in {\cite{Collier:2016cls,Afkhami-Jeddi:2017idc,Benjamin:2019stq}} to argue for accumulation in $h$.

Following Ref.\,{\cite{Collier:2016cls}}, consider
Eq.\,(\ref{modulartransformation}) in the range $0 < \beta, \bar{\beta} <
\infty$, and take the limit $\beta \rightarrow 0$ and $\bar{\beta}
\rightarrow \infty$ (We remark that our $\beta$, $\bar{\beta} $ is what
{\cite{Collier:2016cls}} refers to as $4 \pi^2 / \beta$ and $4 \pi^2 /
\bar{\beta}$. They consider Eq.\,(\ref{Smod}) and do a character expansion,
which is entirely equivalent to studying Eq.\,(\ref{modulartransformation}).)
Keeping $\beta$ fixed, one observes that the vacuum dominates in direct
channel. This is true since each block is $O (e^{- \bar{\beta} T})$ for fixed
$\beta$. At this point, one hopes vacuum dominance survives
\begin{eqnarray}
	&  & \text{- when summed over all blocks;}  \label{hopeAmod}\\
	&  & \text{- not only for fixed $\beta$ but as $\beta \rightarrow 0$, i.e.
		in the modular lightcone limit}.  \label{hopeBmod}
\end{eqnarray}
We point out the similarity of the above hopes with (\ref{hopeA}) and
(\ref{hopeB}). Assuming the vacuum dominance in the direct channel in the modular
lightcone limit,\footnote{In {\cite{Collier:2016cls}}, the first equality in
	Eq.\,(2.11) is the statement of vacuum dominance, true for fixed
	$\beta$ for $\bar{\beta} \rightarrow 0$ (in their convention of $\beta$ and
	$\bar{\beta}$). But then it is assumed to be true in $\beta \rightarrow \infty$
	in the final step of the argument. } one asks how that can be reproduced from
the crossed channel perspective. Looking at the behavior of individual
blocks in crossed channel, one argues a finite number of primaries can never
reproduce it. One further argues that the vacuum dominance can be produced by
infinitely many states with $h$ accumulating to $A$.\footnote{In
	{\cite{Collier:2016cls}}, the second equality in Eq.\,(2.11), assuming it
	remains true in the double lightcone limit, tells us how the vacuum dominance can
	be achieved from the crossed channel perspective.} However, there is no
apriori reasoning behind ruling out the possibility of infinite number of
states with $h$ not accumulating to $A$ in the crossed channel, conspiring to
produce the vacuum dominance in the direct channel. A similar argument is
presented in the context of higher spin CFTs in {\cite{Afkhami-Jeddi:2017idc}}
which bears the weight of similar subtleties.

In Appendix B.1 in {\cite{Kusuki:2018wpa}}, in order to arrive at the twist
accumulation result, a spectral density at large spin is defined (Eq.\,(B.5)),
followed by a swappability condition as implicitly assumed while writing
Eq.\,(B.6). In particular, the limit in $\beta$ variable (in their notation:
again our $\beta, \bar{\beta} $ corresponds to $4 \pi^2 / \beta$, $4 \pi^2 /
\bar{\beta}$ of this paper) is swapped with integral over $\bar{h}$. It
is not at all clear whether such a large spin limit can be defined and a such
swappability condition holds.

Yet another proof appeared in {\cite{Benjamin:2019stq}}, Section 2, which
gives part (a) of Theorem \ref{boundtwist} (but not part (b)). This paper
follows a careful route and we believe that their argument can be turned into a rigorous proof
with a little extra work. For the readers familiar with their argument, the
subtlety lies in assuming that the error term estimate in their Eq.\,(2.13) survives after one multiplies both sides of Eq.\,(2.13) with the factor that they
mention, and takes derivative with respect to $\bar{\beta}$ in the regime
$\bar{\beta} \ll 1 / \beta$.\footnote{To see why this treatment is subtle we
	consider the following toy function $f (\beta) = \sin (e^{k \beta})$ with $k$
	fixed but arbitrarily large. $f (\beta)$ is an $O (1)$ function because it is
	bounded by 1. However, its derivative $f' (\beta) = k e^{k \beta} \cos (e^{k
		\beta})$ is not bounded, indeed it grows
	exponentially fast in $\beta$. This argument tells us that we may not differentiate
	after estimate. We believe the argument of {\cite{Benjamin:2019stq}} can be made rigorous
	after a more careful treatment basically because there are no factors like $f(\beta)$ appearing in the error terms of their Eq.\,(2.13).} To be completely
rigorous, one needs to take the derivative first and then estimate the error.
We remark that our proof does not follow this route and we do not need to take
any derivative.

\section{Similarities and a general formulation}\label{GF}

\subsection{Dictionary}

Above we pointed out several analogies between the arguments of Sections
\ref{LC} and \ref{LMB}. Here we would like to do this more systematically. In
Table \ref{dict} we give a dictionary of the necessary ingredients for the
lightcone problems on the CFT 4pt function and on the 2D CFT torus
partition function. We hope it will give readers a clearer mind why these two
problems are similar.

\renewcommand{\arraystretch}{1.2}
\begin{table}[h]
\centering
	\begin{tabular}{@{}lll@{}}\toprule
		%\hline
		& Lightcone bootstrap$\quad$ & Modular lightcone bootstrap\\
		\midrule
		Variables & $z, \bar{z}$ & $q, \bar{q}$\\
		%\hline
		$'$ map & $z' = 1 - z$ & $q' = e^{4 \pi^2 / \log (q)}$\\
		%\hline
		Crossing factor $K\quad$ & $K (z) = \left( \frac{z}{z' }
		\right)^{\Delta_{\phi}}$ & $K (q) = \sqrt{\frac{2 \pi}{- \log (q)}}
		\frac{1 - q'}{1 - q} \left( \frac{q}{q'} \right)^A$\\
		%\hline
		Index of $K$ & $2 \Delta_{\phi}$ & $A$($\equiv \frac{c - 1}{24}$)\\
		%\hline
		Building block & $g_{\tau, \ell} (z, \bar{z})$ & $\frac{q^h
			\bar{q}^{\bar{h}}}{(1 - q) (1 - \bar{q})}$\\
		%\hline
		Twist & $\tau$ (for both $z$ and $\bar{z}$) & $h$ (for $q$) and $\bar{h}$
		(for $\bar{q}$)\\
		
		Twist gap & $\tau_{\rm gap}$($\leqslant \tau$) & $T$($\leqslant h,
		\bar{h}$)\\
		%\hline
		\bottomrule
	\end{tabular}
	\caption{\label{dict}The dictionary.}
\end{table}

The first column of Table \ref{dict} refers to the 4pt function problem. We
give a half of the crossing factor, the full one being given by
\begin{equation}
	K (z, \bar{z}) = K (z) K (\bar{z}) .
\end{equation}
The index of $K$ is defined as the power of the $z / z'$ factor.

Consider the second column. Here the relevant variables are $q = e^{- \beta}$
and $\bar{q} = e^{- \bar{\beta}}$ which belong to $(0, 1)$, like $z$,
$\bar{z}$. The map $\beta' = 4 \pi^2 / \beta$ transforms in terms of $q$ to
$q' = e^{4 \pi^2 / \log (q)}$.

In Section \ref{LMB} we studied the limit $\beta \rightarrow 0$, $\bar{\beta}
\rightarrow \infty$, which corresponds to $q \rightarrow 1$, $\bar{q}
\rightarrow 0$, while in Section \ref{LC} we considered $z \rightarrow 0$,
$\bar{z} \rightarrow 1$. This interchange of $q, \bar{q}$ w.r.t.~$z, \bar{z}$
is immaterial.

To make the analogy maximal, we factor out $(1 - q) (1 - \bar{q})$ from the
reduced partition function $\eqref{def:Zvir}$. Now the vacuum
contribution becomes 1, like the unit operator contribution in the 4pt case. The nonvacuum building blocks become $\frac{q^h \bar{q}^{\bar{h}}}{(1 -
	q) (1 - \bar{q})}$, as given in Table \ref{dict}.

The building blocks in the modular case have a very simple explicit form. The
``approximate factorization'' which we needed in Section \ref{LC} in the form
of Lemma \ref{AFLz}, holds here exactly. This makes the low twist analysis
easier. For the same reason, Lemma \ref{contMLB}, needed for the continuum
spectrum case, is much easier than Proposition \ref{prop:contsubleading}.

On the other hand, the (half of) the crossing factor is not a pure power in
the modular bootstrap case. This leads to minor complications in the algebra.
Still, the most important factor in the full crossing kernel $K (q, \bar{q})$
is the power $(q / \bar{q}')^A$, since other factors behave logarithmically
more slowly in the $q \rightarrow 0$, $\bar{q} \rightarrow 1$ limit. This
factor determines the index $A$ which sets the twist accumulation point.

This analogy suggests that one should be able to give a completely general
formulation of the arguments of Sections \ref{LC} and \ref{LMB}, which would
imply both theorems proven in those sections. Roughly, this formulation should
proceed as follows. One considers a general crossing equation of the form
\begin{equation}
	f (x, \bar{x}) = K (x) K (\bar{x}) f (x', \bar{x}') \label{abs}
\end{equation}
where $x, \bar{x} \in (0, 1)$, $K (x)$ is a function satisfying
\begin{equation}
	K (x) K (x') = 1, \label{Kk}
\end{equation}
and $' : [0, 1] \rightarrow [0, 1]$ is a continuous map satisfying $0' = 1$,
$1' = 0$, $x'' = x$. The last condition means that the map is an involution.
About $K (x)$ one assumes that it behaves as $x^{\alpha}$ ($\alpha > 0$) as $x
\rightarrow 0$, up to corrections which vary slower than any power. Formally
this means that on any interval $0 < x < x_0 < 1$ and for any $\varepsilon >
0$, we should have a bound
\begin{equation}
	B^{- 1} x^{\alpha + \varepsilon} \leqslant K (x) \leqslant B x^{\alpha
		- \varepsilon} \label{Kbound},
\end{equation}
where $B = B (\varepsilon, x_0) > 0$. The function $f (x, \bar{x})$ is assumed
expandable in building blocks, which satisfy various conditions like the
conformal blocks or the Virasoro characters, including an approximate
factorization property, involving a ``twist'' parameter. Under these very
general conditions, one should be able to prove that there is a family of
building blocks whose ``twist'' accumulates to the crossing kernel index $\alpha$.

In the main text we did not follow this unification strategy for the following
reason. Both lightcone bootstrap and the modular lightcone bootstrap have some
simplifying features with respect to the most general case. These features
are: simple crossing kernel for the lightcone bootstrap case; simple building
blocks in the modular bootstrap case. The most general case will lose these
simplifying features, and will have a proof which is more complicated than
both subcases which are most important in applications.

\subsection{Mapping $q, \bar{q}$ to $z, \bar{z}$}

It is instructive to push the analogy of the previous section a bit further
and find a mapping from the $q$ variable in the previous section to a
$z$ variable so that the crossing transformation $q \rightarrow q'$ becomes
exactly $z \rightarrow 1 - z$. This is possible, and the mapping takes
the following form:
\begin{equation}
	q = \varphi (z),\quad \bar{q} = \varphi (\bar{z}),
	\quad \varphi (w) = e^{- 2 \pi \frac{K (1 - w)}{K
			(w)}}, \label{qzmap}
\end{equation}
where $K (w) \equiv \frac{\pi}{2}  {}_2 F_1 \left( \frac{1}{2}, \frac{1}{2} ;
1 ; w^2 \right)$ is the complete elliptic integral of the first kind. Then the
modular S transformation $\beta \rightarrow 4 \pi^2 / \beta, \bar{\beta}
\rightarrow 4 \pi^2 / \bar{\beta}$ coincides with $z \rightarrow 1 - z,
\bar{z} \rightarrow 1 - \bar{z}$.

This mapping has the following meaning (see e.g.
{\cite{Hartman:2019pcd}}). For a CFT $\mathcal{A}$ with central charge $c$, one
considers the $\mathbb{Z}_2$ symmetric product orbifold of the CFT i.e
$\mathcal{B}=(\mathcal{A} \times \mathcal{A})/\mathbb{Z}_2 $. In this CFT there
are primaries with $h = \bar{h} = c / 16 $ which permute the
two copies of $\mathcal{A}$, called twist fields \cite{Dixon:1986qv}.\footnote{``Twist'' here has
	nothing to do with $\tau = \Delta - \ell$ in the rest of the paper.} The
sphere 4pt function of twist fields in the $\mathcal{B}$ CFT computes the
partition function $Z (\beta, \bar{\beta})$ of the $\mathcal{A}$ CFT, where
$\beta, \bar{\beta}$ depend on the location of the operators on the sphere as in \eqref{qzmap}, up
to a prefactor:
\begin{equation}
	g (z, \bar{z}) = \left( \frac{z \bar{z}}{2^8 \sqrt{(1 - z) (1 - \bar{z})}}
	\right)^{c / 12} Z (\beta, \bar{\beta}) . \label{Ztog}
\end{equation}
Using this relation, the modular invariance condition is equivalent to
crossing symmetry:
\begin{equation}
	Z (\beta, \bar{\beta}) = Z (4 \pi^2 / \beta, 4 \pi^2 / \bar{\beta})
	\quad \Leftrightarrow \quad g (z, \bar{z}) = \left( \frac{z \bar{z}}{(1 - z) (1 - \bar{z})}
	\right)^{c / 8} g (1 - z, 1 - \bar{z}) .
\end{equation}
There is also one-to-one correspondence between the characters of
$\mathcal{A}$ and the $(\text{Virasoro} \times \text{Virasoro})
		/ \mathbb{Z}_2$ conformal blocks of external twist fields of $\mathcal{B}$.

We will discuss the twist accumulation problem for general Virasoro block
expansions in Section \ref{Virasoro}. For the present discussion, let us
forget about the above meaning of the $q$ to $z$ transformation and just view
it as a way to trivialize the crossing involution. Let us view the problem in
the dictionary of the previous section. In the $z$ coordinate, the nonvacuum
building blocks map to
\begin{equation}
	\tilde{g}_{h, \bar{h}} (z, \bar{z}) = \frac{\varphi (z)^h \varphi
		(\bar{z})^{\bar{h}}}{(1 - \varphi (z)) (1 - \varphi (\bar{z}))}
	\quad (h, \bar{h} \neq 0) \label{nonvac-z} .
\end{equation}
It is possible to show that $\varphi (z)$ has a positive series expansion in
$z$, starting from $z^2$:
\begin{equation}
	\varphi (z) = \frac{z^2}{256} + \frac{z^3}{256} + \frac{29 z^4}{8192} +
	\frac{13 z^5}{4096} + \cdots . \label{phi-exp}
\end{equation}
Thus functions {\eqref{nonvac-z}} also have positive series expansion in $z$.
This was one of the properties much used in the lightcone bootstrap analysis
in Section \ref{LC}. The crossing equation becomes:
\begin{equation}
	\tilde{g} (z, \bar{z}) = K (z) K (\bar{z}) \tilde{g} (1 - z, 1 - \bar{z}),
\end{equation}
where $K$ is the crossing factor from the dictionary mapped to the $z$ coordinate, and
\begin{equation}
	\tilde{g} (z, \bar{z}) = 1 + \sum_{h, \bar{h}} \tilde{g}_{h, \bar{h}} (z,
	\bar{z}) .
\end{equation}
It is possible to show that the factor $K (z)$ has asymptotic behavior
\begin{equation}
	K (z) K (\bar{z}) \sim \left( \frac{z}{1 - \bar{z}} \right)^{2 A}
\end{equation}
for $z \rightarrow 0$, $\bar{z} \rightarrow 1$, times factors which vary at
most logarithmically. We could now run an argument similar to Section \ref{LC}
to show that there is an infinite sequence of building blocks with $2 h
\rightarrow 2 A$ and $\bar{h} \rightarrow \infty$ (factor 2 being due to the
expansion {\eqref{phi-exp}} starting from $z^2$). We did not follow this
strategy in the main text, and instead presented the argument in the $\beta, \bar{\beta}$ variables familiar to the modular bootstrap community.

\section{Twist accumulation conjecture in 2D Virasoro
	CFTs}\label{Virasoro}

It is natural to inquire if Theorem \ref{theorem:DLCtwist} has an analogue for
local 2D CFTs of central charge $c > 1$, satisfying a twist gap
assumption in the spectrum of \textit{Virasoro} primaries. This problem was
considered in {\cite{Kusuki:2018wpa}} and {\cite{Collier:2018exn}}, and here
we will review their conclusions.

Let us define the variables
\begin{equation}
	c = 1 + 6 Q^2, \quad Q = b + \frac{1}{b}, \quad h (\alpha) = \alpha (Q -
	\alpha), \quad  \bar{h} (\bar{\alpha}) = \bar{\alpha} (Q -
	\bar{\alpha}) .
\end{equation}
The parameter $b$ is chosen such that if $c \geqslant 25$ we have $0 < b \leqslant
1$ and if we have $25 > c > 1$, $b$ lies on unit circle. We have
\begin{equation}
	0 < h < A \Leftrightarrow 0 < \alpha < \frac{Q}{2}, \qquad h \geqslant A
	\Leftrightarrow \alpha \in \frac{Q}{2} + i \mathbb{R}\,.
\end{equation}
Consider the 4pt correlation function of an operator $\mathcal{O}$. It
can be expressed as a sum of Virasoro blocks:
\begin{equation}
	\langle \mathcal{O} (0) \mathcal{O} (z, \bar{z}) \mathcal{O} (1) \mathcal{O}(\infty) \rangle = \sum_{\alpha_s,
		\bar{\alpha}_s} C_s^2  V_S (\alpha_s) V_S (\bar{\alpha}_s) = \sum_{\alpha_t,
		\bar{\alpha}_t} C_t^2 V_T (\alpha_t) V_T (\bar{\alpha}_t)\,,
\end{equation}
where $V_S (\alpha_s) V_S (\bar{\alpha}_s)$ is the $s$-channel Virasoro
conformal block with exchanged primary having weight $(\alpha_s,
\bar{\alpha}_s)$, $V_T (\alpha_s) V_T (\bar{\alpha}_s)$ is the $t$-channel
Virasoro conformal block with exchanged primary having weight $(\alpha_t,
\bar{\alpha}_t)$.

The fusion kernel $S_{\alpha_s \alpha_t} $ (or $S_{\bar{\alpha}_s
	\bar{\alpha}_t}$) is defined in {\cite{Ponsot:1999uf,Ponsot:2000mt}}
and reviewed in {\cite{Kusuki:2018wpa}} and {\cite{Collier:2018exn}} via
\begin{equation}
	V_T (\alpha_t) = \int_C  \frac{d \alpha_s}{2 \pi i} S_{\alpha_s \alpha_t}
	V_S (\alpha_s), \qquad  V_T (\bar{\alpha}_t) = \int_C  \frac{d \alpha_s}{2
		\pi i} S_{\bar{\alpha}_s \bar{\alpha}_t} V_S (\bar{\alpha}_s)\,, \label{fusion}
\end{equation}
where the $t$-channel Virasoro block (left and right moving part) is expressed
in terms of $s$-channel blocks (left and right moving part). The integral is
defined over some prescribed contour $C$, see {\cite{Collier:2018exn}}. For the
explicit expression for the fusion kernel we refer to {\cite{Kusuki:2018wpa}}
and {\cite{Collier:2018exn}}.

Virasoro Mean Field theory (VMFT) {\cite{Collier:2018exn}} is defined as a
``theory'' with spectrum that we obtain via expressing the $t$-channel vacuum
Virasoro block contribution in terms of $s$-channel data using the fusion kernel,
i.e.~$S_{\alpha_s I} S_{\bar{\alpha}_s I}$. The VMFT has the following properties:
\begin{enumerate}
	\item The operator spectrum is continuous at $h \geqslant A$.
	
	\item For sufficiently light external operators with $\alpha_{\mathcal{O}} < Q / 4$,
	VMFT contains a discrete set of operators with $\alpha_m = 2 \alpha_{\mathcal{O}}  + m b$ such that $\alpha_m < \frac{Q}{2}$. These are absent if $\text{Re} (\alpha_{\mathcal{O}} ) \geqslant \frac{Q}{4}$.
	
	\item The above statements are true with $\alpha \leftrightarrow
	\bar{\alpha} .$
\end{enumerate}
We see that the VMFT spectrum is consistent with Theorem \ref{boundtwist}. It
saturates the upper bound on twist gap (for sufficiently heavy external
operators). Furthermore, $A$ is always a twist accumulation point because of
continuous spectrum above $A$. These properties hold in spite of the fact that
VMFT is not a fully crossing symmetric correlator.

The above conclusions were reached independently and in a similar manner by
Refs.\,{\cite{Kusuki:2018wpa,Collier:2018exn}}, with the term VMFT coined
by {\cite{Collier:2018exn}}.

Refs.\,{\cite{Kusuki:2018wpa,Collier:2018exn}} then pointed out that in any
unitary 2D CFT with $c > 1$, discrete spectrum, and a positive lower bound on
twists of non-vacuum primaries, the OPE spectral density should approach that
of VMFT at large spin. This would be a vast generalization of Theorem~\ref{boundtwist}. Ref.\,{\cite{Collier:2018exn}}, which gives a more detailed
argument, reached this conclusion by showing that in the large spin, fixed
twist limit, the relative importance of non-identity operators in the
$t$-channel is individually suppressed compared to the vacuum contribution if
there is a twist gap, see Eq.\,(3.12) in {\cite{Collier:2018exn}}. We note the
following caveats in deriving the claimed universality:
\begin{enumerate}
	\item First of all, it's not apriori clear that if we take the full
	correlator and write it as a sum of $t$-channel blocks, then we express each
	$t$-channel Virasoro blocks in terms of $s$-channel Virasoro blocks using fusion
	kernel and then sum over these $s$-channel blocks, we get back the correlator
	as required by crossing symmetry. This is specially subtle since the
	spectrum is assumed discrete while each individual $t$-channel block gives
	continuous contribution above $h \geqslant A$ (similarly for $\bar{h}$). In
	particular it is not obvious that the sum over $s$-channel blocks and the
	integral appearing in (\ref{fusion}) is swappable.
	
	\item Even though the individual non-identity $t$-channel blocks yield
	suppressed $s$-channel contributions, it is not guaranteed that the contribution
	is still suppressed after we sum all the non-identity contributions up since
	we have an infinite number of Virasoro blocks being exchanged in the $t$-channel.
	This issue is noted in {\cite{Collier:2018exn}} at the beginning of Section
	3.2.3. 
\end{enumerate}

It would be very interesting to prove a theorem which puts the claim of
universality of VMFT from {\cite{Kusuki:2018wpa,Collier:2018exn}} on rigorous
footing. This would require a better handle on Virasoro blocks.

As noted in the previous section, the partition function can be rewritten as a
4pt function of twist fields in the orbifolded CFT, who have dimension $c /
8$. The above Virasoro fusion kernel approach does not apply verbatim to
this correlator, since the orbifolded CFT admits an extended symmetry
algebra $(\text{Virasoro}{\times}\text{Virasoro})/ \mathbb{Z}_{2}$. Still, it is conceivable that appropriate
fusion kernel for $(\text{Virasoro}{\times}\text{Virasoro})/ \mathbb{Z}_{2}$ would lead to the same
twist accumulation result.\footnote{To complete the list of
		analogies, we note that Theorem
		\ref{boundtwist} may also be argued by expressing the vacuum character in the direct channel of
		the partition function as an integral of crossed-channel characters, the
		spectral density in this integral representation being a nonzero function at
		$h, \bar{h} > (c - 1) / 24$, see {\cite{Kusuki:2018wpa}}, App. B.1.}
		
\begin{remark}
	The universality of the VMFT as well as Theorem~\ref{boundtwist} proven in this paper are contingent upon the existence of a finite twist gap in the spectrum of Virasoro primaries. A perhaps natural expectation is that this assumption is rather generic and there should be many 2D CFTs with $c>1$ satisfying it. Yet, embarrassingly, no explicit example is currently known for which this can be rigorously shown. The most well-studied class of 2D CFTs are \emph{rational} theories (RCFTs), which contain finitely many primaries of the holomorphic and anti-holomorphic vertex operator algebra. All $c>1$ RCFTs admit higher spin symmetries beyond Virasoro, and thus do not satisfy the twist gap assumption. The knowledge of irrational CFTs is limited, and known examples also have higher spin currents, hence zero twist gap. To cite \cite{Yin:2017yyn}, ``either we are missing important constraints $[\ldots]$, or we have no clue what the generic 2D CFT looks like.'' 
	
	A reasonable, unproven, conjecture is that the critical coupled Potts model \cite{Dotsenko:1998gyp} furnishes an example of a $c>1$ 2D CFT satisfying the twist gap assumption.
	
	Recently some progress has been made on this problem by considering $N>4$ copies of unitary minimal models with $c=1-\frac{6}{m(m-1)}$ \cite{Antunes:2022vtb}.  In the undeformed theory,  many higher spin currents exist.  Now this parent CFT is deformed by an appropriate set of operators, which initiate a renormalization group flow. In the large $m$ limit, this flow is perturbative and \cite{Antunes:2022vtb} showed that the deformed theory flows to an interacting CFT with the property that higher spin currents of the original undeformed theory,  up to spin 10,  get lifted. This suggests that the CFT may be irrational with a finite twist gap in the spectrum of Virasoro primaries \cite{Antunes:2022vtb}.  It is still an open problem to how to access arbitrarily large spins and verify the full validity of this conjecture, while remaining within the large $m$ perturbative regime. 
	
	\end{remark}

\section{Conclusions}

Many results in quantum field theory are extremely deep and
	will likely require highly nontrivial mathematics to be put on rigorous footing, perhaps in some distant future. In
	this regard, conformal field theory is in a better shape. Its axioms are
	simple: they amount to saying that observables of the theory, such as
	correlation functions and the torus partition function, have series expansions
	in well-defined building blocks (conformal blocks or characters), with a
	well-understood region of convergence. Thus for CFTs one may hope, and demand,
	that results about the structure of the theory obtained at an intuitive level
	be raised to the status of mathematical theorems, since this likely requires
	only a small amount of extra work. The main purpose of this paper was to
	realize this hope for the lightcone bootstrap problems for the CFT 4pt
	functions and for the torus partition functions of 2D CFTs. These two
problems share a highly similar mathematical structure and lead
to similar ``twist-accumulation'' results.

On the CFT 4pt function side, we proved the prediction made in
{\cite{Fitzpatrick:2012yx,Komargodski:2012ek}} that in any unitary CFT, if the
OPE spectrum of two identical scalar operators ($\phi \times \phi$) has
a twist gap, the OPE contains a family of quasiprimary operators with spins tending
to $+ \infty$ and twists accumulating to $2 \Delta_{\phi}$, where
$\Delta_{\phi}$ is the scaling dimension of the scalar operator $\phi$.
Our proof relies on a crucial approximate factorization property of conformal
blocks, which follows from Hogervorst's dimensional reduction formula
{\cite{Hogervorst:2016hal}}.

On the torus partition function side, we proved that in any Virasoro
invariant, modular invariant, unitary 2D CFT with a normalizable vacuum and
central charge $c > 1$, a twist gap in the spectrum of Virasoro primary
operators implies a family of Virasoro primary operators $\mathcal{O}_{h,
	\bar{h}}$ with $h \rightarrow + \infty$ and $\bar{h} \rightarrow \frac{c -
	1}{24}$ (the same is true with $h$ and $\bar{h}$ interchanged). Our work
justified this prediction originally made in {\cite{Collier:2016cls}} and also in
{\cite{Afkhami-Jeddi:2017idc}}, {\cite{Benjamin:2019stq}}.

An obvious step forward would be to generalize our results to the CFT
4pt functions of non-identical scalar operators and spinning operators,
as predicted in {\cite{Komargodski:2012ek}} and discussed in Section
\ref{non-identical}.

We note that our analysis did not take advantage of the analyticity of CFT
4pt functions and their conformal blocks, which is essential for a
stronger result of the lightcone bootstrap: the OPE of $\phi \times \phi$
also contains a tower of families of quasiprimaries $\mathcal{O}_{\tau,
	\ell}$ with $\ell \rightarrow + \infty$ and $\tau \rightarrow 2
\Delta_{\phi} + 2 n$. In other words, it is expected that there are
infinitely many accumulation twists of quasiprimary operators, given by $2
\Delta_{\phi} + 2 n$ {\cite{Fitzpatrick:2012yx}}. Such families of quasiprimary
operators are exactly the spectrum of mean field CFTs (MFT). There is little
doubt in these lightcone bootstrap results since they are in agreement with
the numerical bootstrap studies of CFT models
{\cite{Simmons-Duffin:2016wlq,Liu:2020tpf,Caron-Huot:2020ouj}}, and with
(nonrigorous) arguments based on the Lorentzian inversion formula
{\cite{Caron-Huot:2017vep,Simmons-Duffin:2017nub}}. It would be nice to
provide a mathematical proof of the MFT universality at large spin.

The lightcone bootstrap also leads to predictions about the OPE spectral
density of the operators whose twist asymptotes to $2 \Delta_{\phi} + 2 n$
{\cite{Fitzpatrick:2012yx,Komargodski:2012ek}}. It would be interesting to
prove those results rigorously, using e.g.~such techniques such as the
Tauberian theory
{\cite{Pappadopulo:2012jk,Qiao:2017xif,Mukhametzhanov:2018zja}}.

Similar to the case of CFT 4pt function, an interesting avenue is to
expound more on the spectral density of the torus partition function along the
lines of
{\cite{Mukhametzhanov:2019pzy,Ganguly:2019ksp,Pal:2019zzr,Pal:2019yhz,Mukhametzhanov:2020swe}}
to derive a Cardy like formula. In particular, it would be interesting to
revisit Ref.\,{\cite{Benjamin:2019stq}} where they argued for a stronger upper
bound $\frac{c - 1}{16}$ on twist gap by using Cardy-like formula in the
lightcone limit. Furthermore, in 2D, the CFT 4pt function can be
decomposed in terms of Virasoro blocks, the VMFT {\cite{Collier:2018exn}}
playing the role of MFT. As discussed in Section \ref{Virasoro}, Refs.\,{\cite{Kusuki:2018wpa}} and
{\cite{Collier:2018exn}} argued that there exists an infinite number of Regge
trajectories with $h \rightarrow \frac{c - 1}{24}$ at large spin and there are
infinite number of operators with spectral density approaching the VMFT in any
given interval of twist above $\frac{c - 1}{12}$. It would be nice to put VMFT universality on firm
footing just like the MFT universality for higher-dimensional CFTs. Finally,
one can consider further refinements of universality such as consider 2D
CFTs with a discrete global symmetry $G$ and refine our theorem along the lines of
{\cite{Pal:2020wwd}}. The Hilbert space of states for such a CFT decomposes
into various sectors, labeled by $\rho$, the irreducible representations of
$G$. One can hope to prove a twist accumulation result in each $\rho$
sector by assuming suitable twist gap conditions. The new ingredient in the
story would be the partition function twisted by a group element $g \in G$,
and its modular transformed avatar, the partition function of the defect
Hilbert space.

We leave these problems for future studies.

\

\section*{Acknowledgments}

Ideas of Quentin Lamouret {\cite{Lamouret}} played a role at the genesis of this
project. We thank Ying-Hsuan Lin, Dalimil Maz{\'a}{\v c}, Balt van Rees and
Sylvain Ribault for useful discussions. When the project got initiated, SP was at the Institute for
Advanced Study and he acknowledges a debt of gratitude for the funding
provided by Tomislav and Vesna Kundic as well as the support from the grant
DE-SC0009988 from the U.S. Department of Energy. SP also acknowledges the
support by the U.S. Department of Energy, Office of Science, Office of High
Energy Physics, under Award Number DE-SC0011632  and by the Walter Burke
Institute for Theoretical Physics.  JQ is grateful for the funding provided by
{\'E}cole Normale Sup{\'e}rieure - PSL and the Simons Bootstrap
Collaboration when this project was initiated
in the workshop ``Bootstrap 2022'' at Porto University. JQ is supported by the Swiss National Science Foundation through the National Centre of Competence in Research SwissMAP and by the Simons Collaboration on Confinement and QCD Strings. SR is supported by the Simons
Foundation grant 733758 (Simons Bootstrap Collaboration). 

\appendix
	
	\section{Approximate factorization of conformal
		blocks}\label{appendix:proofAFL}
	
	In this appendix we will prove Lemma \ref{AFLz}. We will find it convenient to
	prove a slightly different version of the lemma:
	
	\begin{lemma}[Approximate factorization, alternative version]
		\label{lemma:AFL}Let $a, b \in (0, 1)$, $a < b$. Let $d\geqslant 2$ and assume that $\Delta, \ell$ satisfy
		the unitarity bounds. If $\ell = 0$, assume in addition that $\Delta > \nu =
		\frac{d - 2}{2} .$ Then the conformal block $g_{\tau, \ell} (z, \bar{z})$
		satisfies the following two-sided bound
		\begin{equation}
			1 \leqslant \frac{g_{\tau, \ell} (z, \bar{z})}{\frac{(\nu)_{\ell}}{(2
					\nu)_{\ell}} k_{\tau} (\bar{z}) k_{\tau + 2 \ell} (z)} \leqslant K_d (a,
			b) \left( 1 + \frac{\theta (d \geqslant 3)}{\Delta - \nu} \right) 
			\qquad (\forall 0 \leqslant \bar{z} \leqslant a <
			b \leqslant z < 1), \label{CB:sandwich1}
		\end{equation}
		where $K_d (a, b)$ is a finite constant independent of $\tau, \ell$ and of
		$z, \bar{z}$ in the shown range.
	\end{lemma}
	
	Let us quickly see why this version implies Lemma \ref{AFLz}. We have an
	identity:
	\begin{gather}
		k_{\beta} (z) = (4 \rho)^{\beta / 2} G_{\beta} (\rho), \label{kG}\\
		G_{\beta} (\rho) \assign {}_2 F_1 \left( \tfrac{1}{2}, \tfrac{\beta}{2} ;
		\tfrac{\beta + 1}{2} ; \rho^2 \right), \quad \rho = \frac{z}{\left( 1 +
			\sqrt{1 - z} \right)^2} .
	\end{gather}
	We also have $1 \leqslant G_{\beta} (\rho) \leqslant
	\frac{1}{\sqrt{1 - \rho^2}}$ (Lemma \ref{Gprop} below). Hence the denominator
	of {\eqref{CB:sandwich1}} and the denominator of {\eqref{CB:sandwich}} are
	related to each other by
	\begin{equation}
				1 \leqslant \frac{\frac{(\nu)_{\ell}}{(2 \nu)_{\ell}} k_{\tau} (\bar{z})
					k_{\tau + 2 \ell} (z)}{\bar{\rho}^{\tau / 2} F_{\tau, \ell} (z)} \leqslant C
				(a)\,,
			\end{equation}
			where $C (a)$ does not depend on $\tau$ or $\ell$. 
	
	We will first prove the 2D case (App. \ref{2dcase}), then give an inductive
	proof in $d \geqslant 3$ dimensions (App. \ref{dgeq3}), using Hogervorst's
	dimensional reduction formula {\cite{Hogervorst:2016hal}}, which expands
	$d$-dimensional blocks in terms of $(d - 1)$-dimensional ones (App.
	\ref{Hogervorst}).
	
	\begin{remark}
		Following the arguments below carefully, it's easy to see that $K_d (a, b) =
		1 + O (a)$ for $a \rightarrow 0$ and $b$ fixed. 
	\end{remark}
	
	\begin{remark}
		There is another possible way of proving Lemma \ref{lemma:AFL} but we have
		not managed to realize it. In any dimension, one can expand the conformal
		block into collinear blocks:
		\begin{equation}
			g_{\Delta, \ell} (z, \bar{z}) = \underset{n = 0}{\overset{\infty}{\sum}}
			\underset{k = - n}{\overset{n}{\sum}} c_{n, k}  \bar{z}^{\tau / 2 + n}
			k_{\Delta + \ell + 2 k} (z) .
		\end{equation}
		Here the coefficients $c_{n, k}$ are positive for a group-theoretical
		reason: we decompose the $\text{SO} (2, d)$ representation into irreducible
		representations of the subgroup $\text{SL} (2 ; \mathbb{R})$ (which
		preserves the $z$-axis), and $c_{n, k}$ is the sum of squared inner
		products:
		\begin{equation}
			c_{n, k} = \sum\limits_{\psi}  \lvert \langle \psi \lvert \phi (1) \phi (0)
			\rangle \rvert^2, 
		\end{equation}
		where the sum is over an orthonormal basis of collinear primary states
		$\psi$ with collinear twist $h = \frac{\Delta + \ell}{2} + k$ and scaling
		dimension $\Delta + n + k$.
		
		Lemma \ref{lemma:AFL} would follow if we had an upper bound on $c_{n, k}$
		implying the inequality
		\begin{equation}
			\begin{split}
				\underset{n = 0}{\overset{\infty}{\sum}} \underset{k = -
					n}{\overset{n}{\sum}}& \frac{c_{n, k}}{c_{0, 0}}  \bar{z}^n
				\frac{k_{\Delta + \ell + 2 k} (z)}{k_{\Delta + \ell} (z)} \leqslant K_d
				(a, b) \left( 1 + \frac{\theta (d \geqslant 3)}{\Delta - \nu} \right) \\
				 &\left( \forall 0 \leqslant \bar{z}
				\leqslant a < b \leqslant z < 1,\quad \tau
				\leqslant \tau_{\max} \right) .  \\
			\end{split}
		\end{equation}
		Such a bound can be easily derived in 2D, but we do not know how to derive
		it in $d \geqslant 3$. There have been papers deriving explicit
		expressions for the coefficients $c_{n, k}$ (or coefficients of similar
		expansions where $z$ is replaced by $u$ and $\bar{z}$ by $1 - v$)
		{\cite{Simmons-Duffin:2016wlq,Caron-Huot:2017vep,Li:2019cwm,Caron-Huot:2020ouj,Li:2020ijq}},
		but we did not manage to apply them. We leave realization of this strategy as
		an open problem.
	\end{remark}
	
	\subsection{$d = 2$}\label{2dcase}
	
	Recall that the ratio $(\nu)_{\ell} / (2 \nu)_{\ell}$ is defined for $\nu = 0$
	as in footnote \ref{nuratio}. In our normalization (\ref{CB:normalization}),
	the 2D conformal blocks are given by
	\begin{equation}
		g^{(2)}_{\tau, \ell} (z, \bar{z})  = \begin{cases}
			k_{\tau} (z) k_{\tau} (\bar{z})
			& (\ell = 0) \\
			\frac{1}{2} [k_{\tau} (\bar{z}) k_{\tau + 2 \ell} (z) + k_{\tau} (z)
			k_{\tau + 2 \ell} (\bar{z})] & (\ell \geqslant 1).
		\end{cases}
	\end{equation}
	
	For $\ell = 0$, the upper and lower bounds are obvious with constants 1 and in
	the full range of $z, \bar{z}$ (in $d = 2$ there is no pole $1 / (\Delta -
	\nu)$).
	
	Consider then the case $\ell \geqslant 1$. Since $k_{\beta} (z) > 0$ for $z
	\in (0, 1)$, $\beta \geqslant 0$, the lower bound $g^{(2)}_{\tau, \ell} (z,
	\bar{z}) \geqslant {\frac{1}{2}} k_{\tau} (\bar{z}) k_{\tau + 2 \ell} (z)$
	follows by dropping the second term in the block.
	
	For the upper bound, let us write
	\begin{equation}
		g^{(2)}_{\tau, \ell} (z, \bar{z}) = \frac{1}{2} k_{\tau + 2 \ell} (z)
		k_{\tau} (\bar{z}) [1 + X], \label{cb2d:rewrite} \quad X = \frac{k_{\tau + 2
				\ell} (\bar{z}) k_{\tau} (z)}{k_{\tau + 2 \ell} (z) k_{\tau} (\bar{z})} .
	\end{equation}
	We will use identity {\eqref{kG}}. We have the following lemma about
	$G_{\beta} (\rho)$:
	
	\begin{lemma}
		\label{Gprop}For any $0 \leqslant \beta_1 \leqslant \beta_2 < \infty$ we
		have
		\begin{equation}
			1 \leqslant G_{\beta_1} (\rho) \leqslant G_{\beta_2} (\rho) \leqslant
			\frac{1}{\sqrt{1 - \rho^2}} \, .
		\end{equation}
	\end{lemma}
	
	\begin{proof}
		$G_{\beta_1} (\rho) \leqslant G_{\beta_2} (\rho)$ follows from the
		monotonicity of coefficients in the series expansion of $_2 F_1$. We also
		have $G_0 (\rho) = 1$ and $G_{\infty} (\rho) = \frac{1}{\sqrt{1 - \rho^2}}$.
	\end{proof}
	
	Using {\eqref{kG}}, the ratio $X$ is rewritten as
	\begin{equation}
	 X = \left( \frac{\bar{\rho}}{\rho} \right)^{\ell} \frac{G_{\tau}
		(\rho)}{G_{\tau + 2 \ell} (\rho)} \frac{G_{\tau + 2 \ell}
		(\bar{\rho})}{G_{\tau} (\bar{\rho})} \leqslant 1 \cdot 1 \cdot
	\frac{1}{\sqrt{1 - \rho (a)^2}}\,, 
	\end{equation}
using $\bar{\rho} \leqslant \rho$, and Lemma
		\ref{Gprop} for the other two factors. Hence $X$ is uniformly
		bounded by $C (a)$. This implies the upper bound with $C_2 (a, b) = 1 + C
		(a)$. The $d = 2$ case is complete.
	
	\subsection{Dimensional reduction formula}\label{Hogervorst}
	
	We will need Hogervorst's dimensional reduction formula. We will state and use
	this result using blocks $g_{\Delta, \ell}$ labeled by the dimension and
	spin, not by the twist and spin as in the rest of the paper. Hopefully no
	confusion will arise.
	
	\begin{theorem}[Hogervorst {\cite{Hogervorst:2016hal}}, Eq.\,(2.24)]
		Conformal blocks in $d$ dimensions can be expanded in a series of conformal
		blocks in $d - 1$ dimensions, as
		\begin{equation}
			g_{\Delta, \ell}^{(d)} (z, \bar{z}) = \underset{n =
				0}{\overset{\infty}{\sum}} \underset{p = 0}{\sum^{[\ell / 2]}}
			\mathcal{A}_{n, \ell - 2 p}^{(d)} (\Delta, l) g^{(d - 1)}_{\Delta + 2 n,
				\ell - 2 p} (z, \bar{z}), \label{eq:dimreduction}
		\end{equation}
		where ($j = \ell - 2 p$)
		\begin{eqnarray}
			\mathcal{A}^{(d)}_{n, j} (\Delta, \ell) & = & Z_{\ell, j}^{(d)} 
			\frac{\left( \frac{1}{2} \right)_n}{2^{4 n} n!}  \frac{\left( \frac{\Delta
					+ j}{2} \right)_n \left( \frac{\Delta - 2 \nu - j + 1}{2}
				\right)_n}{\left( \frac{\Delta + j - 1}{2} \right)_n \left( \frac{\Delta -
					2 \nu - j}{2} \right)_n} \nonumber\\
			&  & \times \frac{(\Delta - 1)_{2 n} \left( \frac{\Delta + \ell}{2}
				\right)_n \left( \frac{\Delta - \ell - 2 \nu}{2} \right)_n}{(\Delta -
				\nu)_n \left( \Delta - \nu - \frac{1}{2} + n \right)_n \left( \frac{\Delta
					+ \ell + 1}{2} \right)_n \left( \frac{\Delta - \ell - 2 \nu + 1}{2}
				\right)_n},  \label{def:A}\\
			Z_{\ell, j}^{(d)} & = & \frac{\left( \frac{1}{2} \right)_p \ell !}{p! j!} 
			\frac{(\nu)_{j + p}  (2 \nu - 1)_j}{\left( \nu - \frac{1}{2} \right)_{j +
					p + 1}  (2 \nu)_{\ell}}  (j + \nu - 1 / 2) .  \label{def:Z}
		\end{eqnarray}
	\end{theorem}
	
	Some remarks are due here. First, the existence of an expansion like
	{\eqref{eq:dimreduction}} can be inferred very generally from decomposing a
	representation of the $d$-dimensional conformal block under the $(d -
	1)$-dimensional conformal group. It also follows that the coefficients
	$\mathcal{A}^{(d)}_{n, j} (\Delta, \ell)$ have to be nonnegative if $(\Delta,
	\ell)$ satisfies the $d$-dimensional unitarity bounds. These condition are
	verified by the above expressions. Moreover, these coefficients are regular
	(no poles) for $\ell = 0, \Delta > \nu$ and $\ell \geqslant 1$, $\Delta
	\geqslant \ell + 2 \nu$.
	
	Hogervorst {\cite{Hogervorst:2016hal}} derived recursion relations satisfied
	by the coefficients $\mathcal{A}^{(d)}_{n, j} (\Delta, \ell)$. He then guessed
	{\eqref{def:A}} as the (unique) solution of those recursion relations. He then
	could check up to very high recursion order that {\eqref{def:A}} indeed
	remains compatible with the recursion. The complete proof that {\eqref{def:A}}
	is a solution of Hogervorst's recursion relations is so far missing, but there
	is hardly any doubt that they are correct.
	
	It may perhaps be possible to prove {\eqref{def:A}} by comparing
	with the pole structure of Zamolodchikov-like recursion relations
	{\cite{Kos:2013tga}} which have been established rigorously
	{\cite{Penedones:2015aga}}.
	
	\subsection{$d \geqslant 3$ lower bound}\label{dgeq3}
	
	We will argue by induction using (\ref{eq:dimreduction}). The needed lower
	bound is obtained by keeping the term $n = 0, j = \ell$ and throwing out the
	rest (recall that (\ref{eq:dimreduction}) is a positive sum). We have
	\begin{equation}
		g_{\Delta, \ell}^{(d)} (z, \bar{z}) \geqslant \mathcal{A}_{0, \ell}^{(d)}
		(\Delta, \ell) g^{(d - 1)}_{\Delta, \ell} (z, \bar{z}) .
	\end{equation}
	From (\ref{def:A}), it can be seen that:
	\begin{equation}
		\mathcal{A}_{0, \ell}^{(d)} (\Delta, \ell) \frac{(\nu - 1 / 2)_{\ell}}{(2
			\nu - 1)_{\ell}} = \frac{(\nu)_{\ell}}{(2 \nu)_{\ell}}\,.
	\end{equation}
	Hence by induction we get the lower bound
	\begin{equation}
		g_{\tau, \ell}^{(d)} (z, \bar{z}) \geqslant \frac{(\nu)_{\ell}}{(2
			\nu)_{\ell}} k_{\tau} (\bar{z}) k_{\tau + 2 \ell} (z) \label{cb:lowerbound}
		.
	\end{equation}
	
	\subsection{$d \geqslant 3$ upper bound}
	
	Here we will prove the upper bound in $d \geqslant 3$
	\begin{equation}
		\frac{g_{\tau, \ell} (z, \bar{z})}{\frac{(\nu)_{\ell}}{(2 \nu)_{\ell}}
		k_{\tau} (\bar{z}) k_{\tau + 2 \ell} (z)} \leqslant K_d (a, b) \left( 1 +
	\frac{1}{\Delta - \nu} \right)\qquad (\forall 0
	\leqslant \bar{z} \leqslant a < b \leqslant z < 1) \,.
	\end{equation}
	Note that the pole $\Delta = \nu$ can be approached only by the scalar blocks.
	Thus for $\ell \geqslant 1$ the bound says that the blocks are uniformly
	bounded.
	
	As discussed in Section \ref{dgeq3}, the first term in the reduction formula,
	$\mathcal{A}_{0, \ell}^{(d)} (\Delta, \ell) g^{(d - 1)}_{\Delta, \ell} (z,
	\bar{z})$, gives upon induction exactly the main term in the bound we want to
	obtain. Let $Y$ be the sum of all the other terms, i.e.
	\begin{equation}
		Y = \underset{(n, p) \neq (0, 0)}{\overset{}{\sum}} \mathcal{A}_{n, \ell - 2
			p}^{(d)} (\Delta, \ell) g^{(d - 1)}_{\Delta + 2 n, \ell - 2 p} (z, \bar{z})
		.
	\end{equation}
	To show the induction step, we need to show a bound
	\begin{equation}
		\frac{Y}{\mathcal{A}_{0, \ell}^{(d)} (\Delta, \ell) g^{(d - 1)}_{\Delta,
				\ell} (z, \bar{z})} \leqslant C (a, b, d) \left( 1 + \frac{1}{\Delta -
			\nu} \right) \label{Yratio}
	\end{equation}
	with $C (a, b, d)$ independent of $\ell, \Delta$.
	
	Using {\eqref{def:A}} we have
	\begin{eqnarray}
		\frac{\mathcal{A}_{n, \ell - 2 p}^{(d)} (\Delta, \ell)}{\mathcal{A}_{0,
				\ell}^{(d)} (\Delta, \ell)} & = & \frac{(1 / 2)_p \ell !}{p! (\ell - 2 p) !}
		\frac{(\nu)_{\ell - p}  (\nu - 1 / 2)_{\ell}  (2 \nu - 1)_{\ell - 2 p}  (j
			+ \nu - 1 / 2)}{(\nu)_{\ell}  (\nu - 1 / 2)_{\ell - p + 1} (2 \nu -
			1)_{\ell}} \times \frac{(1 / 2)_n}{2^{4 n} n!} \nonumber\\
		&  & \times \frac{\left( \frac{\Delta + j}{2} \right)_n}{\left(
			\frac{\Delta + j - 1}{2} \right)_n } \frac{\left( \Delta - \nu - \frac{1}{2}
			\right)_n}{(\Delta - \nu)_n} \frac{(\Delta - 1)_{2 n}}{\left( \Delta - \nu -
			\frac{1}{2} \right)_{2 n}}  \nonumber\\
		&  & \times \frac{\left( \frac{\Delta + \ell}{2} \right)_n}{\left(
			\frac{\Delta + \ell + 1}{2} \right)_n} \times \frac{\left( \frac{\Delta - 2
				\nu - j + 1}{2} \right)_n}{\left( \frac{\Delta - j - 2 \nu}{2} \right)_n} 
		\frac{\left( \frac{\Delta - \ell - 2 \nu}{2} \right)_n}{\left( \frac{\Delta
				- 2 \nu - \ell + 1}{2} \right)_n}  \label{A-Arat}
	\end{eqnarray}
	In the third line we bound $\left( \frac{\Delta + \ell}{2} \right)_n \left/
	\left( \frac{\Delta + \ell + 1}{2} \right)_n \right.$ by 1 and we also use
	\begin{equation}
		\frac{\left( \frac{\Delta - 2 \nu - j + 1}{2} \right)_n}{\left(
			\frac{\Delta - j - 2 \nu}{2} \right)_n}  \frac{\left( \frac{\Delta - \ell -
				2 \nu}{2} \right)_n}{\left( \frac{\Delta - 2 \nu - \ell + 1}{2} \right)_n}
		\leqslant 1,
	\end{equation}
	(actually $= 1$ for $\ell = 0$, while to see this for $\ell \geqslant 1$ is
	easy using the unitarity bound $\Delta - \ell - 2 \nu \geqslant 0$).
	
	We next study the second line in {\eqref{A-Arat}}, denote it $F_n$. We need
	to pay attention to the pole at $\Delta = \nu$ and spurious poles at $\Delta =
	1$ and $\Delta = \nu + 1 / 2$. For this reason we consider $n = 0, 1$ and $n
	\geqslant 2$ separately. We have
	\begin{align}
		F_0 &= 1,\\
		F_1 &= \frac{(\Delta + j) (\Delta - 1) \Delta}{(\Delta + j - 1) (\Delta -
			\nu) \left( \Delta - \nu + \frac{1}{2} \right)} .
	\end{align}
	Considering separately the cases $\ell = 0$, $\ell = 1$ and $\ell \geqslant 2$
	and using the unitarity bounds in each case, it's easy to see that
	\begin{equation}
		F_1 \leqslant C \left( 1 + \frac{1}{\Delta - \nu} \right)
	\end{equation}
	with $C$ uniform in $\Delta, \ell$, $j$.
	
	Finally for $n \geqslant 2$ we can write
	\begin{equation}
		F_n = F_1 \frac{\left( \frac{\Delta + j}{2} + 1 \right)_{n - 1}}{\left(
		\frac{\Delta + j - 1}{2} + 1 \right)_{n - 1} } \frac{\left( \Delta - \nu +
		\frac{1}{2} \right)_{n - 1}}{(\Delta - \nu + 1)_{n - 1}} \frac{(\Delta)_{2
			n - 2}}{\left( \Delta - \nu + \frac{1}{2} \right)_{2 n - 2}} \,.
		\end{equation}
	We can estimate the first factor by\footnote{To estimate Pochhammer ratios in
		this equation and the next one, the following lemma is useful: \textit{Let
			$x, y > 0$ and let $k$ be a nonnegative integer such that $y + k \geqslant x$.
			Then $ (x)_n / (y)_n \leqslant (1 + n / y)^k$ for all $n \geqslant 0$.} For a
		proof, rewrite $(x)_n / (y)_n$ as $\frac{(x)_n}{(y + k)_n} \times \frac{(y +
			k)_n}{(y)_n}$. The first factor is $\leqslant 1$. For the second factor we use
		the identity $\frac{(y + k)_n}{(y)_n} = \frac{\Gamma (y + k + n) \Gamma
			(y)}{\Gamma (y + k) \Gamma (y + n)} = \frac{(y + n)_k}{(y)_k}$. The last
		expression is $\leqslant (1 + n / y)^k$. }
	\begin{equation}
		\frac{\left( \frac{\Delta + j}{2} + 1 \right)_{n - 1}}{\left( \frac{\Delta +
				j - 1}{2} + 1 \right)_{n - 1} } \leqslant 1 + 2 n \qquad(\Delta \geqslant 0),
	\end{equation}
	the second one by 1, and the last one by
	\begin{equation}
		\frac{(\Delta)_{2 n - 2}}{\left( \Delta - \nu + \frac{1}{2} \right)_{2 n -
				2}} \leqslant (1 + 4 n)^{[\nu]} \qquad (\Delta \geqslant \nu) .
	\end{equation}
	We conclude that for all $n$, $\Delta, \ell, j$ we have
	\begin{equation}
		F_n \leqslant C \left( 1 + \frac{1}{\Delta - \nu} \right) (1 + 4 n)^{[\nu] +
			1} .
	\end{equation}
	Using the discussed estimates for the second and third line in
	{\eqref{A-Arat}}, we get an upper bound for the ratios of the $\mathcal{A}$
	coefficients in {\eqref{Yratio}}:
	\begin{eqnarray}
		\frac{\mathcal{A}_{n, \ell - 2 p}^{(d)} (\Delta, \ell)}{\mathcal{A}_{0,
				\ell}^{(d)} (\Delta, \ell)} & \leqslant & \frac{(1 / 2)_p \ell ! (\nu - 1 /
			2)_{\ell}}{p! (\ell - 2 p) ! (2 \nu - 1)_{\ell}}  \frac{(\nu)_{\ell - p}  (2
			\nu - 1)_{\ell - 2 p}  (j + \nu - 1 / 2)}{(\nu)_{\ell}  (\nu - 1 / 2)_{\ell
				- p + 1}} \nonumber\\
		&  & \times \frac{(1 / 2)_n}{2^{4 n} n!} C \left( 1 + \frac{1}{\Delta -
			\nu} \right) (1 + 4 n)^{[\nu] + 1} .  \label{Aratiobound}
	\end{eqnarray}

	For the ratios of $(d - 1)$-dimensional conformal blocks in {\eqref{Yratio}}
	we use the bounds from the previous induction step (the upper bound in the
	numerator and the lower bound in the denominator). We thus get
	\begin{equation}
		\frac{g^{(d - 1)}_{\Delta + 2 n, \ell - 2 p} (z, \bar{z})}{g^{(d -
				1)}_{\Delta, \ell} (z, \bar{z})} \leqslant C' \frac{(2
			\nu - 1)_{\ell}  (\nu - 1 / 2)_{\ell - 2 p}}{(\nu - 1 / 2)_{\ell}  (2 \nu -
			1)_{\ell - 2 p}}  \frac{k_{\tau + 2 n + 2 p} (\bar{z}) k_{\Delta + j + 2 n}
			(z)}{k_{\tau} (\bar{z}) k_{\Delta + \ell} (z)} . \label{gratio}
	\end{equation}
	where $C' = K_{d - 1} (a, b) \left( 1 + \frac{\theta (d - 1 \geqslant 3)
	}{\Delta + 2 n - \nu_{d - 1}} \right)$ is unformly bounded in
	$\Delta, \ell$ since $\nu_{d - 1} = \nu - 1 / 2$.
	
	For further simplification of {\eqref{gratio}} we use the identity
	{\eqref{kG}}. In addition to Lemma \ref{Gprop}, we will need another lemma
	about $G_{\beta}$:
	
	\begin{lemma}
		\label{lemma:Gratiobound}For $\beta_1 \geqslant \beta_2 > 0$ we have
		\begin{equation}
			\frac{G_{\beta_1} (\rho)}{G_{\beta_2} (\rho)} \leqslant \frac{\Gamma
				\left( \frac{\beta_1 + 1}{2} \right) \Gamma \left( \frac{\beta_2}{2}
				\right)}{\Gamma \left( \frac{\beta_2 + 1}{2} \right) \Gamma \left(
				\frac{\beta_1}{2} \right)}\qquad (0 \leqslant
			\rho < 1) .
		\end{equation}
	\end{lemma}
	
	\begin{proof}
		We use the integral representation of ${}_2 F_1$ valid for $\text{Re} (c) > \text{Re} (b) > 0$, $\lvert \arg
		(1 - x) \rvert < \pi$:
		\begin{gather}
				{}_2 F_1 (a, b ; c ; x) = \frac{\Gamma (c)}{\Gamma (b) \Gamma (c - b)}
				\int_0^1 \frac{t^{b - 1} (1 - t)^{c - b - 1}}{(1 - x t)^a} d t 
				%\nonumber\\(\text{Re} (c) > \text{Re} (b) > 0,\quad \lvert \arg(1 - x) \rvert < \pi) . 
			\end{gather}
		Taking $a = \frac{1}{2}, b = \frac{\beta}{2}, c = \frac{\beta + 1}{2}$ and
		$x = \rho^2$ we get
		\begin{eqnarray}
			G_{\beta_1} (\rho) & = & \frac{\Gamma \left( \frac{\beta_1 + 1}{2}
				\right)}{\Gamma \left( \frac{1}{2} \right) \Gamma \left( \frac{\beta_1}{2}
				\right)} \int_0^1 \frac{t^{\frac{\beta_1}{2} - 1} (1 - t)^{-
					\frac{1}{2}}}{(1 - \rho^2 t)^{\frac{1}{2}}} d t \, . 
		\end{eqnarray}
		We replace $t^{\beta_1 / 2}$ under the integral by a larger quantity $t^{\beta_2 / 2}$, and infer the stated bound.
	\end{proof}
	
	By {\eqref{kG}}, we rewrite the ratio of $k$'s in {\eqref{gratio}} as
	\begin{equation}
		\frac{k_{\tau + 2 n + 2 p} (\bar{z}) k_{\Delta + j + 2 n} (z)}{k_{\tau}
			(\bar{z}) k_{\Delta + \ell} (z)} = (4 \bar{\rho})^{n + p} \frac{G_{\tau + 2
				n + 2 p} (\bar{\rho})}{G_{\tau} (\bar{\rho})} (4 \rho)^{n - p}
		\frac{G_{\Delta + 2 n + j} (\rho)}{G_{\Delta + \ell} (\rho)} .
	\end{equation}
	By Lemma \ref{Gprop}, we have
	\begin{equation}
		\frac{G_{\tau + 2 n + 2 p} (\bar{\rho})}{G_{\tau} (\bar{\rho})} \leqslant
		G_{\tau + 2 n + 2 p} (\bar{\rho}) \leqslant \frac{1}{\sqrt{1 -
				\bar{\rho}^2}} \leqslant \frac{1}{\sqrt{1 - \rho (a)^2}} .
	\end{equation}
	By Lemmas \ref{Gprop} and \ref{lemma:Gratiobound} we have
	\begin{equation}
		\frac{G_{\Delta + 2 n + j} (\rho)}{G_{\Delta + \ell} (\rho)} \leqslant
		\frac{G_{\Delta + 2 n + j} (\rho)}{G_{\Delta + j} (\rho)} \leqslant
		\frac{\left( \frac{\Delta + j + 1}{2} \right)_n}{\left( \frac{\Delta + j}{2}
			\right)_n} \leqslant 1 + 4 n \qquad (\Delta \geqslant \nu) .
	\end{equation}
	Plugging these estimates into {\eqref{gratio}}, we estimate it as
	\begin{equation}
		\frac{g^{(d - 1)}_{\Delta + 2 n, \ell - 2 p} (z, \bar{z})}{g^{(d -
				1)}_{\Delta, \ell} (z, \bar{z})} \leqslant C'' (a, b, d)
		\frac{(2 \nu - 1)_{\ell}  (\nu - 1 / 2)_{\ell - 2 p}}{(\nu - 1 / 2)_{\ell} 
			(2 \nu - 1)_{\ell - 2 p}}  (16 \rho \bar{\rho})^n \left(
		\frac{\bar{\rho}}{\rho} \right)^p (1 + 4 n) . \label{gratiobound}
	\end{equation}
	Combining (\ref{Aratiobound}) and (\ref{gratiobound}), and using the
	inequality
	\begin{equation}
		\frac{\ell !}{(\ell - 2 p) !}  \frac{(\nu)_{\ell - p}  (\nu - 1 / 2)_{\ell -
				2 p + 1}}{(\nu)_{\ell}  (\nu - 1 / 2)_{\ell - p + 1}} \leqslant 2,
	\end{equation}
	we get
	\begin{eqnarray}
		\frac{\mathcal{A}_{n, \ell - 2 p}^{(d)} (\Delta, \ell) g^{(d - 1)}_{\Delta +
				2 n, \ell - 2 p} (z, \bar{z})}{\mathcal{A}_{0, \ell}^{(d)} (\Delta, \ell)
			g^{(d - 1)}_{\Delta, \ell} (z, \bar{z})} & \leqslant & C''' (a, b, d) \left(
		1 + \frac{1}{\Delta - \nu} \right) \frac{(1 / 2)_p }{p!}  \left(
		\frac{\bar{\rho}}{\rho} \right)^p \nonumber\\
		&  & \times \frac{(1 / 2)_n}{n!}  (1 + 4 n)^{[\nu] + 2}  (\rho
		\bar{\rho})^n .  \label{termbnd}
	\end{eqnarray}
	The point of this bound is that it only depends on $n, p, \rho, \bar{\rho}$
	but is uniform in $\Delta, \ell$. Also the dependence on $n, p$ is factorized.
	
	We now sum {\eqref{termbnd}} over $p, n$ to get a bound on the ratio
	{\eqref{Yratio}}. Although $p$ in Hogervorst's formula is bounded by $[\ell /
	2]$ we will extend the sum over all $p$ from $0$ to $\infty$. The sum over $p$
	is convergent when $\bar{\rho} / \rho < 1$ while the sum over $n$ is
	convergent when $\rho \bar{\rho} < 1$. Since we are assuming that $0 \leqslant
	\bar{z} \leqslant a < b \leqslant z < 1$, we have
	\begin{equation}
		\quad \rho \bar{\rho} \leqslant \rho (a) < 1, \quad
		\bar{\rho} / \rho \leqslant \rho (a) / \rho (b) < 1 .
	\end{equation}
	So the sums over $n$ and over $p$ are both bounded by a constant which only
	depends on $a, b$ and $d$. We thus obtain the needed bound {\eqref{Yratio}}.
	This completes the induction step.

	\section{Proof of Lemma
		\ref{lemma:improvedlog}}\label{app:improvedlog}
	
	Here we will prove Lemma \ref{lemma:improvedlog}, i.e.\,the improved
	logarithmic bound on the conformal blocks:
	\begin{equation}
		\frac{g_{\tau, \ell} (z, \bar{z})}{g_{\tau, \ell} (b, a)} \leqslant B
		\bar{z}^{\tau / 2} \log \left( \frac{1}{1 - z} \right) \qquad (0 \leqslant \bar{z} \leqslant a < b \leqslant z < 1)
		\label{toprove}
	\end{equation}
	for all $(\tau, \ell)$ satisfying $\tau \leqslant \tau_{\max}, \ell \leqslant
	\ell_{\max}$ and the unitarity bounds, i.e.\,$\tau \geqslant \nu$ for $\ell =
	0$ and $\tau \geqslant 2 \nu$ for $\ell \geqslant 1$. We will show that the
	constant factor $B$ in the upper bound is finite and depends only on $a$, $b$,
	$\tau_{\max}$ and $\ell_{\max}$.
	
	We will divide the proof into two cases: (a) $1 \leqslant \ell \leqslant
	\ell_{\max}$, (b) $\ell = 0$. We will derive (\ref{toprove}) with constants
	$B_1$ for (a) and $B_2 $ for (b). In the end we take $B = \max \{ B_1, B_2
	\}$.
	
	\subsection{$1 \leqslant \ell \leqslant \ell_{\max}$}
	
In this case the range of $\tau$ we are concerned about is given by
	\begin{equation}
		2 \nu \leqslant \tau \leqslant \tau_{\max} .
	\end{equation}
	As explained in Remark \ref{AFLrem}(c), Lemma \ref{AFLz} implies a
	two-sided bound for $2 \nu \leqslant \tau \leqslant \tau_{\max}$
	\begin{equation}
		C_1 \leqslant \frac{g_{\tau, \ell} (z, \bar{z})}{ \bar{z}^{\tau / 2}
			F_{\tau, \ell} (z)} \leqslant C_2 \qquad (0 < \bar{z} \leqslant a < b
		\leqslant z < 1),
	\end{equation}
	where $F_{\tau, \ell}$ is given in {\eqref{def:Ftaul}}, the constants $C_1$,
	$C_2$ are finite and they only depend on $a$, $b$ and $\tau_{\max}$. Using the
	upper bound for $g_{\tau, \ell} (z, \bar{z})$ and the lower bound for
	$g_{\tau, \ell} (b, a)$, we get
	\begin{equation}
		\begin{split}
			&\frac{g_{\tau, \ell} (z, \bar{z})}{g_{\tau, \ell} (b, a)} \leqslant
			\frac{C_2  \bar{z}^{\tau / 2} F_{\tau, \ell} (z)}{C_1 a^{\tau / 2} F_{\tau,
					\ell} (b)} = \frac{C_2  \bar{z}^{\tau / 2} k_{\tau + 2 \ell} (z)}{C_1
				a^{\tau / 2} k_{\tau + 2 \ell} (b)} \\
			&\left( 0
			< \bar{z} \leqslant a < b \leqslant z < 1,\quad 2
			\nu \leqslant \tau \leqslant \tau_{\max} \right), \\
		\end{split}
	\end{equation}
	here we used the explicit form of $F_{\tau, \ell}$ in the last step. Then
	using the lower bounds $a^{\tau / 2} \geqslant a^{\tau_{\max} / 2}$, $k_{\tau
		+ 2 \ell} (b) \geqslant b^{\tau_{\max} / 2 + \ell_{\max}}$ for the
	denominator, and the upper bounds $z^{\beta / 2} \leqslant 1$, $k_{\beta} (z)
	\leqslant {}_2 F_1 (\beta / 2, \beta / 2 ; \beta ; z) \leqslant {}_2 F_1
	(\beta_{\max} / 2, \beta_{\max} / 2 ; \beta_{\max} ; z)$ for the numerator, we
	get
	\begin{equation}
		\frac{g_{\tau, \ell} (z, \bar{z})}{g_{\tau, \ell} (b, a)} \leqslant C'
		\bar{z}^{\tau / 2} {}_2 F_1 \left( \frac{\tau_{\max} + 2 \ell_{\max}}{2},
		\frac{\tau_{\max} + 2 \ell_{\max}}{2} ; \tau_{\max} + 2 \ell_{\max} ; z
		\right) \label{ratiobound1}
	\end{equation}
	in the regime where $0 < \bar{z} \leqslant a < b \leqslant z < 1$, $2 \nu
	\leqslant \tau \leqslant \tau_{\max}$ and $1 \leqslant \ell \leqslant
	\ell_{\max}$. Here the coefficient $C' \assign C_2 / (C_1 a^{\tau_{\max} / 2}
	b^{\tau_{\max} / 2 + \ell_{\max}})$ is finite and depends only on $a$, $b$,
	$\tau_{\max}$ and $\ell_{\max}$. For the hypergeometric function ${}_2 F_1$ we
	use the fact (mentioned in Remark \ref{AFLrem}(e)) that
	\begin{equation}
		{}_2 F_1 \left( \frac{\beta}{2}, \frac{\beta}{2} ; \beta ; z \right) \sim
		\frac{\Gamma \left( \frac{\beta + 1}{2} \right) 2^{\beta - 1}}{\Gamma \left(
			\frac{\beta}{2} \right) \sqrt{\pi}} \log \left( \frac{1}{1 - z} \right)
		\qquad (z \rightarrow 1) .
	\end{equation}
	This asymptotics together with the monotonicity of ${}_2 F_1 \left(
	\frac{\beta}{2}, \frac{\beta}{2} ; \beta ; z \right)$ imply that
	\begin{equation}
		{}_2 F_1 \left( \frac{\beta}{2}, \frac{\beta}{2} ; \beta ; z \right)
		\leqslant C'' \log \left( \frac{1}{1 - z} \right) \qquad \left( b \leqslant z < 1,\quad \forall
		\beta \geqslant 0 \right), \label{ratiobound2}
	\end{equation}
	where $C''<\infty$ only depends on $0<\beta<\infty$ and $b>0$. Combining (\ref{ratiobound1}) and
	(\ref{ratiobound2}) we get
	\begin{equation}
		\frac{g_{\tau, \ell} (z, \bar{z})}{g_{\tau, \ell} (b, a)} \leqslant B_1
		\bar{z}^{\tau / 2} \log \left( \frac{1}{1 - z} \right)
	\end{equation}
	in the regime $0 < \bar{z} \leqslant a < b \leqslant z < 1$, $2 \nu \leqslant
	\tau \leqslant \tau_{\max}$ and $1 \leqslant \ell \leqslant \ell_{\max}$. Here
	$B_1 : = C' C''$ is finite and depends only on $a$, $b$,
	$\tau_{\max}$ and $\ell_{\max}$. This finishes the proof of $\ell \geqslant 1$
	case.
	
	\subsection{$\ell = 0$}
	
	The $d = 2$, $\ell = 0$ case can be shown exactly as $\ell \geqslant 1$ in the
	previous section, since in this case the $\Delta \rightarrow \nu$ pole is
	absent. In the rest of this section we consider $d \geqslant 3$.
	
By upper bound in Lemma \ref{AFLz} and $\bar{\rho} \leqslant \bar{z}$
	we get the upper bound of $g_{\Delta, 0} (z, \bar{z})$ for $\Delta > \nu$:
	\begin{equation}
		g_{\Delta, 0} (z, \bar{z}) \leqslant K_d (a, b) \left( 1 + \frac{1}{\Delta -
			\nu} \right) \bar{z}^{\Delta / 2} F_{\Delta, 0} (z)  \qquad (0 \leqslant \bar{z} \leqslant a < b \leqslant z < 1),
	\end{equation}
	where $K_d (a, b)$ is a finite constant independent of $\Delta$. Then by the
	similar analysis as the $\ell \geqslant 1$ case, we get an upper bound on
	$g_{\Delta, 0} (z, \bar{z})$:
	\begin{equation}
		\begin{split}
			&g_{\Delta, 0} (z, \bar{z}) \leqslant D_1 \left( 1 + \frac{1}{\Delta - \nu}
			\right) \bar{z}^{\Delta / 2} \log \left( \frac{1}{1 - z} \right) \\
			&\left( 0 \leqslant \bar{z} \leqslant a < b
			\leqslant z < 1,\quad \nu < \Delta \leqslant
			\tau_{\max} \right), \\
		\end{split}
	\end{equation}
	where $D_1$ is a finite constant, depending only on $a$, $b$ and
	$\tau_{\max}$. The main subtlety here is the $(\Delta - \nu)^{- 1}$ term which
	blows up when $\Delta \rightarrow \nu$. To show that the ratio $g_{\Delta, 0}
	(z, \bar{z}) / g_{\Delta, 0} (b, a)$ has the $\bar{z}^{\tau / 2} \log \left(
	\frac{1}{1 - \bar{z}} \right)$ bound with the coefficient uniform in $\Delta$,
	we need a lower bound of $g_{\Delta, 0} (b, a)$ which also blows up as
	$(\Delta - \nu)^{- 1}$ when $\Delta \rightarrow \nu$. The lower bound in Lemma
	\ref{AFLz} does not show this behavior. So we need an improved lower bound of
	$g_{\Delta, 0} (b, a)$ which features the $(\Delta - \nu)^{- 1}$ singularity.
	
	For this we use the dimensional reduction formula
	(\ref{eq:dimreduction}), restricting to $\ell = 0$:
	\begin{equation}
		g_{\Delta, 0}^{(d)} (z, \bar{z}) = \underset{n = 0}{\overset{\infty}{\sum}}
		\mathcal{A}_{n, 0}^{(d)} (\Delta, 0) g^{(d - 1)}_{\Delta + 2 n, 0} (z,
		\bar{z}), \label{dimred:scalar}
	\end{equation}
	where the reduction coefficients are given by
	\begin{equation}
		\mathcal{A}^{(d)}_{n, 0} (\Delta, 0) = \frac{\left( \frac{1}{2} \right)_n
			\left( \left( \frac{\Delta}{2} \right)_n \right)^3}{4^n n! (\Delta - \nu)_n
			\left( \Delta - \nu - \frac{1}{2} + n \right)_n  \left( \frac{\Delta + 1}{2}
			\right)_n} . \label{dimredcoeff:scalar}
	\end{equation}
	When $\Delta > \nu$, (\ref{dimred:scalar}) is a positive sum, so $g_{\Delta,
		0}^{(d)}$ is bounded from below by the $n = 1$ term:
	\begin{equation}
		g_{\Delta, 0}^{(d)} (z, \bar{z}) \geqslant \frac{\left( \frac{1}{2} \right)
			\left( \frac{\Delta}{2} \right)^3}{4 (\Delta - \nu) \left( \Delta - \nu +
			\frac{1}{2} \right)  \left( \frac{\Delta + 1}{2} \right)} g^{(d -
			1)}_{\Delta + 2, 0} (z, \bar{z}) .
	\end{equation}
	This estimate gives us the needed lower bound of $g_{\Delta, 0}^{(d)} (b, a)$
	since it has the $(\Delta - \nu)^{- 1}$ factor. The whole prefactor in front
	of $g^{(d - 1)}_{\Delta + 2, 0} (z, \bar{z})$ is bounded from below by
	\begin{equation}
		\frac{\left( \frac{1}{2} \right) \left( \frac{\Delta}{2} \right)^3}{4
			(\Delta - \nu) \left( \Delta - \nu + \frac{1}{2} \right)  \left(
			\frac{\Delta + 1}{2} \right)} \geqslant \frac{D_2}{\Delta - \nu}
		\qquad (\nu < \Delta \leqslant \tau_{\max}),
	\end{equation}
	where $D_2 \assign \frac{\nu^3}{32 \left( \tau_{\max} - \nu + \frac{1}{2}
		\right) (\tau_{\max} + 1)} > 0$. Now take $z = b$ and $\bar{z} = a$, the
	estimate of $g^{(d - 1)}_{\Delta + 2, 0} (b, a)$ is similar to the previous
	section:
	\begin{eqnarray}
		g^{(d - 1)}_{\Delta + 2, 0} (b, a) & \geqslant & 2^{- \tau_{\max} - 2}
		a^{(\Delta + 2) / 2} F_{\Delta + 2, 0} (b)\\
		& = & 2^{- \tau_{\max} - 2} a^{(\Delta + 2) / 2} 4^{(\Delta + 2) / 2}
		b^{\Delta + 2} {}_2 F_1 (\Delta + 2, \Delta + 2 ; 2 \Delta + 4 ; b)
		\nonumber\\
		& \geqslant & 2^{- \tau_{\max} - 2} a^{(\tau_{\max} + 2) / 2} 4^{(\nu + 2)
			/ 2} b^{\tau_{\max} + 2} {}_2 F_1 (\nu + 2, \nu + 2 ; 2 \nu + 4 ; b)=: D_3
		\nonumber 
	\end{eqnarray}
	Here in the first line we used the lower bound of Lemma \ref{AFLz} and the fact
	that $\bar{\rho}^{(\Delta + 2) / 2} \geqslant (\bar{z} / 4)^{(\Delta + 2) / 2}
	\geqslant 2^{- \tau_{\max} - 2} \bar{z}^{(\Delta + 2) / 2}$ for $\Delta
	\leqslant \tau_{\max}$, in the second line we used the explicit form of
	$F_{\Delta + 2, 0} (b)$, in the third line we bounded $a^{(\Delta + 2) / 2}$,
	$4^{(\Delta + 2) / 2}$, $b^{\Delta + 2}$ and ${}_2 F_1$ by their minimal values in the range $\nu \leqslant
	\Delta \leqslant \tau_{\max}$. $D_3>0$ because $a, b \in (0, 1)$ and $\tau_{\max} < \infty$.
	
	Putting these together, we have for $0 \leqslant \bar{z} \leqslant a < b
	\leqslant z < 1$ and $\nu < \Delta \leqslant \tau_{\max}$:
	\begin{equation}
		\frac{g_{\Delta, 0} (z, \bar{z})}{g_{\Delta, 0} (b, a)} \leqslant
		\frac{D_1}{D_2 D_3} (\Delta - \nu) \left( 1 + \frac{1}{\Delta - \nu} \right)
		\bar{z}^{\Delta / 2} \log \left( \frac{1}{1 - z} \right) \leqslant B_2
		\bar{z}^{\Delta / 2} \log \left( \frac{1}{1 - z} \right),
	\end{equation}
	where $B_2 \assign \frac{D_1}{D_2 D_3} (1 + \tau_{\max} - \nu)$ is finite and
	depends only on $a$, $b$ and $\tau_{\max}$.

	\small

\bibliography{lightcone}
\bibliographystyle{utphys}

\end{document}